\documentclass{article}
\usepackage[utf8]{inputenc}
\usepackage{amsmath,amsthm,amssymb}
\usepackage{mathtools}
\usepackage{bm,bbm}
\usepackage{xifthen}
\usepackage{enumitem}
\usepackage{authblk}
\usepackage[top=1in, bottom=1in, left=1.25in, right=1.25in]{geometry}

\usepackage{letltxmacro}
\makeatletter
\let\oldr@@t\r@@t
\def\r@@t#1#2{%
\setbox0=\hbox{$\oldr@@t#1{#2\,}$}\dimen0=\ht0
\advance\dimen0-0.2\ht0
\setbox2=\hbox{\vrule height\ht0 depth -\dimen0}%
{\box0\lower0.4pt\box2}}
\LetLtxMacro{\oldsqrt}{\sqrt}
\renewcommand*{\sqrt}[2][\ ]{\oldsqrt[#1]{#2}}
\makeatother

\newcommand{\pder}[2]{
    \ifthenelse{\isempty{#1}}
        {\frac{\partial}{\partial#2}}
        {\frac{\partial#1}{\partial#2}}
}
\newcommand{\der}[2]{
    {\ifthenelse{\isempty{#1}}
        {\frac{\mathrm{d}}{\mathrm{d}#2}}
        {\frac{\mathrm{d}#1}{\mathrm{d}#2}}
    }
}

\let\temp\epsilon
\let\epsilon\varepsilon
\let\varepsilon\temp

\let\temp\phi
\let\phi\varphi
\let\varphi\phi

\newcommand{\E}{\mathcal{E}}

\renewcommand{\k}{\widehat{k}}
\newcommand{\TT}{\text{TT}}
\newcommand{\norm}[1]{\lVert#1\rVert}
\newcommand{\hnorm}[2]{\norm{#1}_{H^{#2}}}
\newcommand{\lnorm}[2]{\norm{#1}_{L^{#2}}}
\newcommand{\Id}{\operatorname{Id}}
\newcommand{\tr}{\operatorname{tr}_g}

\newtheorem{theorem}{Theorem}[section]
\newtheorem{corollary}{Corollary}[theorem]
\newtheorem{lemma}[theorem]{Lemma}
\newtheorem{proposition}[theorem]{Proposition}
\newtheorem*{remark}{Remark}

\title{On the Global Well-Posedness of the Einstein-Yang-Mills System}
\author[1,2]{Petar Griggs\thanks{\ttfamily pmg912@college.harard.edu}}
\author[1,3]{Puskar Mondal\thanks{\ttfamily puskar\_mondal@fas.harvard.edu}}
\affil[1]{Department of Mathematics, Harvard University}
\affil[2]{Department of Physics, Harvard University}
\affil[3]{Center of Mathematical Sciences and Applications, Harvard University}

\begin{document}

\maketitle

\begin{abstract}
\noindent 
In this paper, we present a partial result on the global well-posedness of the Cauchy problem for the Einstein-Yang-Mills system in the constant mean extrinsic curvature spatial harmonic and generalized Coulomb gauges as introduced in \cite{Mondal2021}. We give a small-data global existence theorem for a family of $n+1$ dimensional spacetimes with $n\geq4$, utilizing energy arguments presented in \cite{Andersson2009}. We observe that these energy arguments will fail for $n=3$ due to the conformal invariance of the $3+1$ Yang-Mills equations and present a gauge-covaraiant formulation of the Einstein-Yang-Mills system in $3+1$ dimensions to show that an energy argument cannot be used to prove the global well-posedness result, regardless of the choice of gauge.
\end{abstract}

\section{Introduction}
One of the fundamental questions in the study of hyperbolic partial differential equations is the issue of long \textit{time} existence. In gauge theories these issues are made complex due to the fact that the choice of gauge plays an important role in the existence problem. As an example, the spacetime harmonic was utilized to prove the local well-posedness of the vacuum Einstein's equations by \cite{ChoquetBruhat1995}. This gauge was believed to be unstable for large time. However, \cite{Lindblad2010} later proved the global existence of the Minkowski space in spacetime harmonic gauge for small data perturbations, suggesting otherwise. Another familiar example is the use of the Lorentz gauge in the Maxwell theory. Use of these gauges cast the equations of motion into a hyperbolic system whose local well-posedness theory follows using standard methods. Long time existence issues however are substantially more complicated in the case of non-linear equations and require careful attention to the detail of the non-linearities present. In several cases choice of gauge plays an important role, in the sense that while one gauge develops coordinate singularity in finite time, other gauges may be able to exhaust the entire time interval. In addition to these gauge choices which are naturally adapted to equations in fully spacetime covariant form, there are other gauges suitable for studying the equations on a spacetime $M$ of the product type $\mathbb{R}\times\Sigma$. One of the most convenient choices is the constant mean curvature spatial harmonic gauge (CMCSH gauge) introduced by Andersson and Moncrief \cite{Andersson2003} to study the vacuum Einstein equations. Application of CMCSH gauge turns Einstein's equations into a coupled elliptic-hyperbolic system, whose local well-posedness was proven in \cite{Andersson2003}. In the context of Yang-Mills theory, the Coulomb gauge is a natural analog of the spatial harmonic gauge, and it was used by Klainerman and Machedon \cite{Klainerman1995} to study the local and global existence problems of the Yang-Mills equations on Minkowski spacetime. 

In addition to the fundamental studies presented in the previous paragraph, numerous studies in the literature deal with local and global existence problems of the Einstein and Yang-Mills equations. However, we will only describe the studies that are relevant to our problem. Firstly, we mention the work of Lindblad and Rodnianski \cite{Lindblad2010}, which used the spacetime harmonic gauge to obtain a global stability result of Minkowski space. LeFloch and Ma \cite{LeFloch2016} used the same gauge to establish a small data global existence result of the Einstein-Klein-Gordon system. In addition to the stability problem of Minkowski space, spacetime harmonic gauge is also used in the context of the cosmological stability problem by \cite{Rodnianski2009,Speck2012}. Beyond these examples, there are numerous other studies, e.g. \cite{Lindblad2017,Ettinger2017,Ames2017}, that utilize this particular choice of gauge. Andersson and Moncrief \cite{Andersson2009} proved an asymptotic stability result of the Milne universe utilizing the CMCSH gauge. Later, Andersson and Fajman used the CMCSH gauge to prove a small data global existence result for the Einstein-Vlasov system on a Milne spacetime \cite{Andersson2020}. In addition, \cite{Wang2014} obtained a rough data local well-posedness result for the vacuum Einstein equations, \cite{Fajman2016} proved a local well-posedness result of the Einstein-Vlasov system, \cite{Fajman2021B} studied the stability of the Milne universe in the presence of a perfect fluid, and \cite{Fajman2021} studied the asymptotic stability of the Milne universe coupled to a Klein-Gordon field, among others. There are numerous studies where CMCSH gauge is used to study the vacuum gravity problem or gravity coupled to matter fields, for example \cite{Moncrief2019,Mondal2020,Mondal2021C}. A substantial amount of study has also been done in the Yang-Mills sector. Eardley and Moncrief \cite{Eardley1982A,Eardley1982B} established the global existence of Yang-Mills fields on a Minkowski background in temporal gauge. Later, \cite{Chrusciel1997} extended the global existence result to globally hyperbolic curved spacetimes in the same gauge. As mentioned earlier, using Coulomb gauge, Klainerman and Machedon \cite{Klainerman1995} established a global existence result in energy norm. Tao \cite{Tao2003} subsequently presented a below-energy-norm local existence theorem for the Yang-Mills equations on Minkowski space in temporal gauge. There are countless other studies which are less relevant in the current context.

Our current article is motivated by the desire to extend the result of \cite{Andersson2009} to include a non-trivial Yang-Mills source. We employ a generalized Coulomb (GC) gauge to cast the Yang-Mills equations into a coupled elliptic-hyperbolic system, mirroring the gravitational sector which is also reduced to a coupled elliptic-hyperbolic system by the application of the CMCSH gauge. A local well-posedness of this coupled elliptic hyperbolic system in CMCSHGC gauge was established in \cite{Mondal2021}. However, we find that the energy argument presented in \cite{Andersson2009} for pure gravity fails for the Einstein-Yang-Mills system in $3+1$ dimensions. In particular, the decay of Yang-Mills fields are \textit{weak} in an appropriate sense and as such fail to overcome the non-linearities present through a straightforward energy argument. This is tied to the conformal invariance of Yang-Mills equations in $3+1$ dimensions: the Milne universe is conformal to a cylinder spacetime that exhibits no decay, and so it is natural that Yang-Mills fields would not have an uniform decay on the $3+1$ Milne background. This problem does not occur in the presence of Maxwell sources, since the later is linear \cite{Branding2019} and uniform decay for the Maxwell field is not required to close the argument. However, in $n+1$ with $n\geq 4$ dimensions, a modified energy type argument works and we obtain the desired asymptotic behaviour that leads to small data global existence result. The result essentially depends on the spectrum of a Lichnerowicz-type Laplacian (acting on $2-$tensors) as well as that of a Hodge-type Laplacian (acting on $1-$forms). The additional dependence on the spectrum of a Hodge-type Laplacian originates from the Yang-Mills sector.

The failure of the energy method in $3+1$ dimensions necessitates applications of sophisticated analytical techniques such as light-cone estimates \cite{Moncrief2022,Vazquez2022}. We believe that the $3+1$ Milne model is stable under coupled Einstein-Yang-Mills perturbations and we are unable to prove it at present using energy estimates alone. We note that Minkowski space is stable under coupled Einstein-Yang-Mills perturbations \cite{Mondal2022B}, and the proof required sophisticated double null energy estimates; such a technique fails in the current context due to non-trivial topology of the spacetime. We intend to address this problem in the future.

The structure of the article is as follows. In sections \ref{prelim} and \ref{rescaling}, we present a brief overview of the background spacetimes we are studying and the construction of the Yang-Mills theory as in \cite{Mondal2021}. We also discuss the gauge-fixing and re-scaling that we must do to study the dynamics over time of the geometric quantities that describe the spatial structure of the spacetime. We describe the small-data scenario of the Einstein-Yang-Mills system.

In section \ref{elliptic}, we give a series of estimates for the gauge variables, while in section \ref{energies} we construct a energy analogous to that presented in \cite{Andersson2009} and find leading-order estimates on the decay of the energies. The decay of the energy gives us a global existence result for the Cauchy initial value problem of the Einstein-Yang-Mills problem. Finally, in section \ref{gauge-covariant}, we present a gauge covariant formulation of the Einstein-Yang-Mills in $3+1$ dimensions to demonstrate that an energy argument of the form used for $n+1$, with $n\geq4$, dimensions will fail.

\section{Preliminaries}\label{prelim}
Here we will briefly review the setup, terminology, and notation that we will use throughout the rest of this work. Take an integer $n>3$. Let $M$ be a globally hyperbolic $n+1$ dimensional spacetime manifold. As $M$ is globally hyperbolic, we may foliate the spacetime by closed, spatial Cauchy hypersurfaces $\Sigma$, such that $M$ is decomposed as $(0,\infty)\times\Sigma$. Let $\mathcal{M}_\Sigma$ denote the space of Riemannian metrics on a constant time hypersurface $\Sigma$, and let $\mathcal{E}_\alpha\subset\mathcal{M}_\Sigma$ be the subset of negative Einstein metrics with Einstein constant $-\alpha$, i.e., the set of all metrics $\gamma$ that satisfy
\begin{equation*}
\mathrm{Ric}[\gamma]=-\alpha\gamma.
\end{equation*}

\noindent  We will restrict our interest in this paper to metrics in the space $\mathcal{E}_\alpha$, where we may re-scale the metric such that $\alpha=\frac{n-1}{n^2}$. If we take a $\gamma\in\mathcal{E}_\alpha$, we then have a \textit{hyperbolic} metric $\tilde{\gamma}$ on $M$ given by
\begin{equation}\label{M-background-metric}
\tilde{\gamma}=-\frac{n^2}{t^4}\mathrm{d}t\otimes\mathrm{d}t+\frac{1}{t^2}{\gamma}_{ij}\:\mathrm{d}x^i\otimes\mathrm{d}x^j
\end{equation}

\noindent for $t\in(0,\infty)$. Observe that $\tilde{\gamma}$ also has vanishing Ricci, and indeed Riemann, curvature and hence is a solution of the vacuum Einstein's equations. This solution will serve as the background geometry in our later analysis, and we will aim to study whether small perturbations to both the geometry and matter source decay to yield again a metric of this form.

Throughout the paper, we will use $g$ to denote the dynamical metric that satisfies Einstein's equations in the ADM formalism. We will denote by $\gamma$ a $C^\infty$ background metric that satisfies the so-called \textit{shadow-metric condition} with respect to $g$; see \ref{shadow-section} for a discussion of the conditions we impose on $\gamma$. Note that we will always have $\gamma\in\mathcal{E}_\alpha$.

We view tensors as sections of vector bundles, and we will frequently make use of index notation when working with tensors. We let Greek indices ($\mu,\nu,\dots$) denote tensors on the spacetime manifold $M$, and let Latin indices from the middle of the alphabet ($i,j,\dots$) denote tensors on the spatial hypersurface $\Sigma$.

Now, if $u$ and $v$ are sections of a rank-$(m,n)$ tensor bundle over $\Sigma$, we define an inner product as
\begin{equation}
\langle u,v\rangle=\int_\Sigma\gamma_{i_1j_1}\dots\gamma_{i_nj_n}\gamma^{k_1\ell_1}\dots\gamma^{k_m\ell_m}u^{i_1\dots i_n}{}_{k_1\dots k_m}v^{j_1\dots j_n}{}_{\ell_1\dots\ell_m}\:\mu_g,
\end{equation}

\noindent where $\mu_g$ is the volume form $\sqrt{\det g}\mathrm{d}x^1\wedge\dots\wedge\mathrm{d}x^n$. We will occasionally abuse this notation and use $\mu_g$ to simply denote $\sqrt{\det g}$, though it will be clear from context when we do or do not include the measure $\mathrm{d}x^1\wedge\dots\wedge\mathrm{d}x^n$.

We will denote the covariant derivative with respect to the metrics $g$ and $\gamma$ as $\nabla$ and $\nabla[\gamma]$, respectively. We will use this type of notation when referring to other metric-derived quantities, such as the Riemann curvature $R$ and $R[\gamma]$, but will often still explicitly specify the Christoffel symbols as $\Gamma[g]$ and $\Gamma[\gamma]$. We will occasionally use $\nabla[\gamma]^i$ to denote $i$-times repeated covariant derivatives.

\noindent For derivatives of the rank-$(m,n)$ tensors $u$ and $v$, we define the inner product
\begin{equation}\label{inner-prod-der}
\langle\nabla[\gamma]u,\nabla[\gamma]v\rangle=\int_\Sigma\gamma_{i_1j_1}\dots\gamma_{i_nj_n}\gamma^{k_1\ell_1}\dots\gamma^{k_m\ell_m}g^{rs}(\nabla[\gamma]_ru)^{i_1\dots i_n}{}_{k_1\dots k_m}(\nabla[\gamma]_sv)^{j_1\dots j_n}{}_{\ell_1\dots\ell_m}\:\mu_g.
\end{equation}

\noindent The inner product for higher derivatives of $u$ and $v$ is defined analogously. Now, define the twisted rough Laplacian as the second-order differential operator
\begin{equation}
\Delta_g^\gamma=-\frac{1}{\sqrt{\mu_g}}\nabla[\gamma]_i(g^{ij}\mu_g\nabla[\gamma]_j).
\end{equation}

\noindent We note that the twisted Laplacian is self-adjoint with respect to the inner product \eqref{inner-prod-der}. We define the usual covariant Laplacian as $\Delta_g=-g^{ij}\nabla_i\nabla_j$, such that $\Delta_g$ has non-negative spectrum. From $\Delta_g^\gamma$, we can define two additional differential operators which will be key to the later analysis. First, let $\mathcal{L}_{g,\gamma}$ denote the second-order differential operator that acts on two-forms $h$ as
\begin{equation}
\mathcal{L}_{g,\gamma}h_{ij}=\Delta^\gamma_gh_{ij}-2g^{mk}g^{n\ell}R[\gamma]_{imjn}h_{k\ell}.
\end{equation}

\noindent Similarly, let $\mathfrak{L}_{g,\gamma}$ denote the second-order differential operator that acts on one-forms $\omega$ as
\begin{equation}
\mathfrak{L}_{g,\gamma}\omega_i=\Delta^\gamma_g\omega_i-g^{j\ell}R[\gamma]^m{}_{\ell ij}\omega_m.
\end{equation}

\noindent We note that $\mathfrak{L}_{\gamma,\gamma}=\Delta_\gamma-\text{Ric}[\gamma]$ is related to the Hodge Laplacian $\Delta_H\equiv(\mathrm{d}\mathrm{d}^\star+\mathrm{d}^\star\mathrm{d})$ by
\begin{equation*}
\mathfrak{L}_{\gamma,\gamma}\omega_i=\Delta_H\omega_i-2\mathrm{Ric}[\gamma]^m{}_{i}\omega_m
\end{equation*}

\noindent and, like the Hodge Laplacian, will also have non-negative spectrum as $\mathrm{Ric}[\gamma]$ is negative-definite. For $g$ near $\gamma$, we will then also have the spectrum of $\mathfrak{L}_{g,\gamma}$ is non-negative. We will also assume that $\mathcal{L}_{g,\gamma}$ has non-negative spectrum for $g$ near $\gamma$, motivated by the fact that there is no known compact, negative Einstein manifold for which $\mathcal{L}_{\gamma,\gamma}$ has negative eigenvalues \cite{Andersson2009}. The assumption of non-negativity of these operators spectra will be key in the later energy arguments.

For the Yang-Mills theory, we recall the setup from \cite{Mondal2021}. Let $P$ be a principle $G$-bundle over $M$ where $G$ is any compact, semi-simple Lie group of dimension $d_G$. Let $\mathfrak{g}$ denote the Lie algebra of $G$ and take $V$ to be a real $d_V$-dimensional representation of $\mathfrak{g}$. Because $G$ is compact, we may find a positive-definite adjoint-invariant metric on $\mathfrak{g}$. We take $\{(\chi_A)^a{}_b\}_{A=1}^{d_G}$ to be a basis of $\mathfrak{g}$ in a $d_V$-dimensional representation $V$ for which
\begin{equation*}
-\operatorname{Tr}(\chi_A\chi_B)=(\chi_A)^a{}_b(\chi_B)^b{}_a=\delta_{AB},
\end{equation*}

\noindent where $a,b\in\{1,\dots,d_V\}$ denote indices at the level of the representation $V$. We will reserve Latin indices from the beginning of the alphabet ($a,b,\dots$) for the indices of the representation $V$. Occasionally we will suppress these indices and simply write the Lie algebra-valued $k$-form $A^a{}_b$ as $A$, where $A$ is then understood to be a $d_V\times d_V$ real matrix. With this choice of metric, we may define a gauge invariant inner product on $\mathfrak{g}$. With $A$ and $B$ both Lie algebra-valued $(m,n)$-tensors, we have an inner product given by
\begin{equation}
\langle A,B\rangle=\int_\Sigma\gamma_{i_1j_1}\dots\gamma_{i_nj_n}\gamma^{k_1\ell_1}\dots\gamma^{k_m\ell_m}A^a{}_b{}^{i_1\dots i_n}{}_{k_1\dots k_m}B^a{}_b{}^{j_1\dots j_n}{}_{\ell_1\dots\ell_m}\:\mu_g.
\end{equation}

\noindent Note that $A$ and $B$ must themselves be gauge covariant objects for the inner product $\langle A,B\rangle$ to be gauge invariant. We may also extend this definition to the derivatives $\nabla[\gamma]A$ and $\nabla[\gamma]B$ in a fashion identical to \eqref{inner-prod-der}.

Finally, we will have a connection $1$-form on $P$, which for the chosen basis and representation we will denote as $\tilde{A}^a{}_b=\tilde{A}^a{}_{b\mu}\mathrm{d}x^\mu$. From $\tilde{A}$, we may define the Yang-Mills field strength as
\begin{equation}
\tilde{F}^a{}_{b\mu\nu}=\partial_\mu\tilde{A}^a{}_{b\nu}-\partial_\nu\tilde{A}^a{}_{b\mu}+\big[\tilde{A},\tilde{A}\big]^a{}_{b\mu\nu},
\end{equation}

\noindent where $[\bullet,\bullet]$ is the Lie bracket of $\mathfrak{g}$ and corresponds to the usual matrix commutator in the representation $V$. Note that we fix the Yang-Mills coupling constant as unity since we are working at a classical level thus do not allow running of the coupling constant. We also note that $\tilde{A}$ is not gauge covariant, though we may construct a gauge covariant quantity by subtracting from $\tilde{A}$ another $\mathfrak{g}$-valued $1$-form. We also emphasize that we will view $\tilde{A}$ as a perturbation of the background geometry, which is a solution to the vacuum Einstein's equations.

Using the various inner products defined above, we may define the $L^2(\Sigma)$ norm of a tensor $v$ by $\lnorm{v}{2}=\langle v,v\rangle^\frac{1}{2}$, where the appropriate inner product is used based on the particular type of object that $v$ is. The $L^2(\Sigma)$ norm will be especially important, as we will conduct our analysis in the Sobolev spaces $H^s$, where we fix $s>\frac{n}{2}+1$. The choice of such an $s$ ensures that we have the embedding $H^{s-1}(\Sigma)\hookrightarrow L^\infty(\Sigma)$, where the $L^\infty$ norm of a generic Lie algebra-value rank-$(m,n)$ tensor $v$ is given by
\begin{equation*}
\lnorm{v}{\infty}=\sup_{\Sigma}\big(\gamma_{i_1j_1}\dots\gamma_{i_nj_n}\gamma^{k_1\ell_1}\dots\gamma^{k_m\ell_m}v^a{}_b{}^{i_1\dots i_n}{}_{k_1\dots k_m}v^b{}_a{}^{j_1\dots j_n}{}_{\ell_1\dots\ell_m}\big)^\frac{1}{2}
\end{equation*}

\noindent and is defined analogously for the other types of objects we have discussed.

\subsection{Gauge Fixing of Evolution Equations}
In this section we will discuss how we construct the Einstein-Yang-Mills system and the gauge choices that we must make to study the evolution of the geometric variables. This is largely an overview of the construction of the Einstein-Yang-Mills system presented in \cite{Mondal2021}.

As we work in the ADM formalism, we foliate $M$ by spacelike hypersurfaces $\Sigma_t$ and decompose the vector field $\partial_t$ on the manifold $M$ into components orthogonal to and parallel to $\Sigma_t$, writing
\begin{equation}
\partial_t=N\mathbf{n}+X.
\end{equation}

\noindent Here the shift vector field $X=X^i\partial_{x_i}$ is a section of the tangent bundle $T\Sigma_t$, while the lapse function $N$ controls the normal distance between two constant $t$ hypersurfaces. From this decomposition, we may write Einstein's equations in the form of evolution equations for the Riemannian metric $g$ and second fundamental form $k$ of $\Sigma_t$. We denote by $\tau=\tr(k)$ the trace of the second fundamental form. The quantity $\tau$ is the mean extrinsic curvature of $\Sigma_t$ as a surface embedded in $M$, and will, under suitable assumptions, provide us with a natural time function. We note that, because the Milne universe model that we are considering is expanding, we have $\tau<0$.

As we will consider a matter source in Einstein's equations given by a Yang-Mills stress-energy tensor, we will also wish to decompose the connection $1$-form $\tilde{A}$ and field strength $\tilde{F}$ into components parallel to and orthogonal to the hypersurface $\Sigma_t$. We do so by writing
\begin{align}
\tilde{A}^a{}_b=A^a{}_b-\tilde{g}(\tilde{A}^a{}_b,\mathbf{n})\mathbf{n}, \\
\tilde{F}^a{}_b=\E^a{}_b\otimes\mathbf{n}-\mathbf{n}\otimes\E^a{}_b+F^a{}_b.
\end{align}

\noindent Here $A^a{}_b$ and $\E^a{}_b$ are Lie algebra-valued one-forms, with $\E^a{}_{bi}$ the components of the electric field. Meanwhile $F^a{}_{bij}=\partial_iA^a{}_{bj}-\partial_jA^a{}_{bi}+[A,A]^a{}_{bij}$ is a Lie algebra-valued two form on $\Sigma_t$ related to the magnetic field. The value $\tilde{g}(\tilde{A}^a{}_b,\partial_t)$ is denoted as $A^a{}_{b0}$. After performing this splitting, we find a system of evolution and constraint equations for the quantities $(k,g,\E,A)$; see equations (26)-(31) in \cite{Mondal2021}. However, to obtain a hyperbolic system of evolution equations, we need to make an appropriate choice of gauge. In particular, the system of evolution and constraint equations leave us free to make choices of $N$, $X$, and a quantity related to $A^a{}_{b0}$.

We will fix the gauge variables $N$ and $X$ that arise from the ADM decomposition using the \textit{constant mean extrinsic curvature spatial harmonic} (CMCSH) gauge. We first set the mean extrinsic curvature $\tau$ to be constant on each leaf in the foliation of $M$, thus enabling $\tau$ to serve as a time function. To make this choice of time function, we must know that the spacetime $M$ admits a constant mean extrinsic curvature hypersurface. As discussed in \cite{Bartnik1988,Rendall1996,Andersson1999,Galloway2018}, not all spacetimes will admit such a hypersurface. However the background metric $\tilde{g}$ will, and so we will assume that the perturbed spacetime we work with does as well. In particular, we will take $t=\tau$. Explicitly computing $\pder{\tau}{t}$ from the definition $\tau=\tr(k)$ then yields an elliptic equation for the lapse function $N$; see Eq. (34) of \cite{Mondal2021}. To fix the shift vector field $X$, we utilize the spatial harmonic gauge, which is the requirement that the tension field vanishes, i.e.,
\begin{equation}
V^k=g^{ij}\big(\Gamma[g]^k{}_{ij}-\Gamma[\gamma]^k{}_{ij}\big)=0.
\end{equation}

\noindent Taking the time derivative of $V^k$ in the spatial harmonic gauge and imposing the constraint on $\nabla\k$ obtained from the Codazzi equation, we may obtain an elliptic equation for $X$; see Eq. (37) of \cite{Mondal2021}. We note that in the spatial harmonic gauge, the twisted rough Laplacian $\Delta^\gamma_g$ will take the simplified form 
\begin{equation}
\Delta^\gamma_g=-g^{mn}\nabla[\gamma]_m\nabla[\gamma]_n+V^m\nabla[\gamma]_m=-g^{mn}\nabla[\gamma]_m\nabla[\gamma]_n.
\end{equation}

\noindent Finally, we must fix a Yang-Mills gauge variable. We define $\phi\in\mathfrak{g}$ to be $\phi^a{}_b=A^a{}_{b0}-A^a{}_{bi}X^i$ to be this gauge variable. We do not directly use the temporal coordinate of the connection $1$-form $A_0$ as it would require a higher regularity than $\phi$ and thus the later energy arguments would not close. To obtain a constraint on $\phi$, we adopt the \textit{generalized Coulomb gauge}, given by the condition
\begin{equation}\label{gen-col-gauge}
g^{ij}\nabla^{\widehat{A}}[\gamma]_i(A_j-\widehat{A}_j)\coloneqq g^{ij}\nabla[\gamma]_i(A_j-\widehat{A}_j)+g^{ij}\big[\widehat{A}_i,(A_j-\widehat{A}_j)\big]=0.
\end{equation}

\noindent Here $\nabla^{\widehat{A}}[\gamma]$ is the gauge covariant derivative with respect to the smooth background connection $1$-form $\widehat{A}$. Taking the time derivative of \eqref{gen-col-gauge} and imposing the constraint on $\nabla\E$ obtained from Gauss' law then yields an elliptic equation on $\phi$. The full set of evolution and constraint equations is given in Section 3 of \cite{Mondal2021}, though we do not present them here as we will, in Section \ref{rescale-eqns-section}, give a scaled form of the equations; see \eqref{A-evol}--\eqref{phi-constraint}.

\subsubsection{Shadow Metric Condition}\label{shadow-section}
If $\gamma_0\in\mathcal{E}_{(n-1)/n^2}$ is the background metric near which we perturb to get the dynamical metric $g$, we will not, in general, have that $g$ also lies in $\mathcal{E}_{(n-1)/n^2}$. However, as we wish to study the stability of $M$, we will need to study whether the dyanmical metric $g$ converges to an element of $\mathcal{E}_{(n-1)/n^2}$. For this, we present the shadow metric condition, introduced in \cite{Andersson2009}.

To begin, we must introduce the concept of the deformation space of a fixed metric $\gamma_0\in\mathcal{E}_{(n-1)/n^2}$. First, define the set $\mathcal{S}_{\gamma_0}$ to those metrics $g\in\mathcal{M}_\Sigma$ such that the identity map $(M,g)\mapsto(M,\gamma_0)$ is harmonic. This is equivalent to the tension field of $g$ and $\gamma_0$ vanishing, i.e.,
\begin{equation*}
\mathcal{S}_{\gamma_0}=\big\{g\in\mathcal{M}_\Sigma\:|\:g^{ij}\big(\Gamma[g]^k_{ij}-\Gamma[\gamma_0]^k{}_{ij})=0\big\}.
\end{equation*}

\noindent Hence $\mathcal{S}_{\gamma_0}$ will be the set of metrics that satisfy the CMCSH gauge condition. We have that $\mathcal{S}_{\gamma_0}$ defined in this way is both a submanifold of the space $\mathcal{M}_\Sigma$ close to $\gamma$ and a slice for the diffeomorphism group, that it is invariant under isometries of $\gamma_0$; c.f. \cite{Andersson2009} and \cite{Besse1987}. With $\mathcal{S}_{\gamma_0}$, we may define the \textit{deformation space} of $\gamma_0$ as follows: let $\mathcal{V}\subset\mathcal{E}_{(n-1)/n^2}$ be the connected component of $\gamma_0$ and $\mathcal{S}_{\gamma_0}$ the slice defined above. The deformation space of $\gamma_0$ is then defined to be the set $\mathcal{N}=\mathcal{V}\cap\mathcal{S}_{\gamma_0}$. Throughout the remainder of this paper we will use $\mathcal{N}$ to refer to the deformation space of a fixed $\gamma_0\in\E_{(n-1)/n^2}$ and $\gamma$ to refer to a general element of $\mathcal{N}$.

The deformation space of $\gamma_0$ is thus the set of background metrics which we expect the dynamical metric $g$ of Einstein's equations to converge to; namely they are the negative Einstein metrics that preserve the CMCSH gauge condition. Note that by Theorem 5.26 of \cite{Besse1987} we may find a set of coordinates such that $\gamma\in\mathcal{E}_{(n-1)/n^2}$ is smooth. It is thus reasonable for us to assume that $\mathcal{N}\subset\mathcal{M}_\Sigma$ has a $C^\infty$ topology so that we may smoothly evolve metrics within the space.

By the discussion in Section 2.3 of \cite{Andersson2009}, we have that the formal tangent space $T_\gamma\mathcal{N}$ for some $\gamma\in\mathcal{N}$ is given by
\begin{equation}\label{N-tan-space}
T_\gamma\mathcal{N}=\{h\in\mathrm{Sym}^2(\Sigma)\:|\:h=h^{\TT\parallel}+L_{Y^\parallel}\gamma\},
\end{equation}

\noindent where $\mathrm{Sym}^2(\Sigma)$ is the set of symmetric, rank-$2$ covariant tensors on $\Sigma$, $h^{\TT\parallel}$ is a transverse-traceless tensor with respect to $\gamma$, and $Y^{\parallel}$ is a vector field that solves
\begin{equation}\label{Y-constraint}
(DV)_j(L_{Y^\parallel}\gamma)^{ij}=h^{\TT\parallel mn}\big(\Gamma[\gamma]^i{}_{mn}-\Gamma[\gamma_0]^i{}_{mn}\big),
\end{equation}

\noindent with $DV$ the Fr\'echet derivative of the tension field $V^k=\gamma^{ij}(\Gamma[\gamma]^k{}_{ij}-\Gamma[\gamma_0]^k{}_{ij})$ taken with respect to $\gamma$. It turns out that $\mathcal{N}$ is a finite-dimensional submanifold of $\mathcal{M}_\Sigma$ \cite{Besse1987}, and so we let $\{q^\alpha\}_{\alpha=1}^{\dim\mathcal{N}}$ be a set of local coordinates of $\mathcal{N}$. If $\gamma\in\mathcal{N}$, then we may construct a basis for the formal tangent space $T_\gamma\mathcal{N}$ as $\big\{\pder{\gamma}{q^\alpha}\big\}_{\alpha=1}^{\dim\mathcal{N}}$. For an arbitrary Riemannian metric $g\in\mathcal{M}_\Sigma$, we can define the $L^2$ projection with respect to the metric $\gamma$ as
\begin{equation}
\bigg\langle g,\pder{\gamma}{q^\alpha}\bigg\rangle_\gamma=-\int_\Sigma\gamma^{im}\gamma^{jn}g_{ij}\pder{\gamma_{mn}}{q^\alpha}\:\mu_\gamma.
\end{equation}

\noindent With this inner product, we can state the shadow metric condition as follows: if $g\in\mathcal{M}$ and $\gamma\in\mathcal{N}$, then $\gamma$ is a shadow metric of $g$ if
\begin{equation}\label{shadow-metric-cond}
\bigg\langle g-\gamma,\pder{\gamma}{q^\alpha}\bigg\rangle_\gamma=0
\end{equation}

\noindent for all $\alpha\in\{1,\dots,\dim\mathcal{N}\}$. By Lemma 4.3 of \cite{Andersson2009}, we have that this shadow metric condition will always be uniquely satisfied by some $\gamma\in\mathcal{N}$ for $g\in\mathcal{M}_\Sigma$ in a sufficiently small neighborhood of $\gamma_0$.

Geometrically we may view the shadow metric condition as being that $u\coloneqq g-\gamma$ is orthogonal to the deformation space $\mathcal{N}$, or equivalently that $\gamma$ is the projection of $g$ onto $\mathcal{N}$. Indeed, the evolution of $g$ will be entirely determined by the horizontal evolution of $\gamma$ and the orthogonal evolution of $u$. As $\gamma$ will evolve within the deformation space $\mathcal{N}$, then to see whether $g$ converges to a point in $\mathcal{N}$ it will suffice to study whether the orthogonal component $u$ vanishes. In light of the view of $\gamma$ as the projection of $g$, we note that the shadow metric condition \eqref{shadow-metric-cond} will give us an estimate of the time evolution $\partial_T\gamma$ of the form
\begin{equation}\label{dT-gamma-est}
\norm{\partial_T\gamma}\leq C\lnorm{\partial_Tg}{2}\leq C\hnorm{\partial_Tg}{s-1}
\end{equation}

\noindent if $g$ lies in a suitably small neighborhood of $\gamma_0$. Here $C>0$ is a constant depending only on the background geometry, and the norm of $\partial_T\gamma$ is arbitrary as $\gamma$, and hence $\partial_T\gamma$, lies in a finite-dimensional vector space. We also note that, if $\mathcal{L}_{\gamma_0,\gamma_0}$ has strictly positive spectrum, then $\mathcal{N}=\{\gamma_0\}$ \cite{Andersson2009}. This will not, however, affect our analysis as we may still trivially bound $\partial_T\gamma=0$ by $C\hnorm{\partial_Tg}{s-1}\geq0$.

Finally, we remark that one may also wish to define a similar condition for the analogous quantity in the Yang-Mills theory, namely the connection one-form $A$. However, unlike the deformation space of metrics, the deformation space of connection one-forms, which we denote $\mathcal{A}$, is not in general well-understood. In fact, as even the topology of $\mathcal{A}$ is unknown, we cannot hope to perform an analogous decomposition to the dynamics of $A$ and study whether $A$ will converge to a non-vacuum solution of the Einstein's equations. Since the vanishing of the Yang-Mills curvature form $F$ will yield a solution to the Einstein-Yang-Mills system, we will perturb the connection one-form $A$ about the fixed flat background connection $\hat{A}\equiv0$ throughout the rest of this work.

\section{Rescaled Einstein's Equations}\label{rescaling}
\subsection{Rescaled Elliptic and Constraint Equations}\label{rescale-eqns-section}
We follow \cite{Andersson2009} in rescaling the system of evolution and constraint equations in the CMCSH and generalized Coulomb gauges to obtain a system of equations describing the behavior of dimensionless quantities. However, unlike the vacuum Einstein equations studied by Andersson and Moncrief, we will find that the Yang-Mills matter source blocks the rescaled system from being fully autonomous. The rescalling will, however, ensure that all factors of $\tau$ only appear alongside the rescaled Yang-Mills variables $\E$ and $F=\mathrm{d}A+[A,A]$. This is a consequence of the Yang-Mills matter field not being scale invariant in $n+1$, with $n\geq 4$, dimensions, and physically can be thought of as the hyperbolic structure of the spacetime dispersing the matter source and causing the Yang-Mills field strength to vanish as $\tau\rightarrow0$. We note, however, that in $3+1$ dimensions, the Yang-Mills equations are conformally invariant.

As in \cite{Andersson2009}, we will take $\tau=\tr(k)$ to have dimensions of $(\text{length})^{-1}$ and the spatial coordinates $x^i$ to be dimensionless. For the quantities derived from the geometry of $\Sigma$, namely $(k,g,N,X)$, we obtain the scaling
\begin{alignat*}{2}
[g]&\sim(\text{length})^2, &\quad\quad [k]&\sim(\text{length}), \\
[N]&\sim(\text{length})^2, & [X]&\sim(\text{length}).
\end{alignat*}

\noindent Because $k$ has an implicit dependence on $\tau$ as its trace, we will want to instead study the evolution of the traceless part of $k$, which we will denote $\k$. Explicitly, we let $\k=k-(\tau/n)g$, such that $\tr(\k)=0$. Observe that we still have $[\k]\sim(\text{length})$.

To determine the scaling of the Yang-Mills variables $(\E,A,\phi)$, we impose the condition that the gauge covariant derivative $\nabla^A[g]=\nabla[g]+[A,\bullet]$ should remain scale invariant. This then gives us the scaling
\begin{equation*}
[A_i]\sim(\text{length})^0, \quad\quad [\E]\sim(\text{length}), \quad\quad [\phi]\sim(\text{length})^{-1}.
\end{equation*}

\noindent We also note that $[F]=[A_i]\sim(\text{length})^0$. Knowing how each of the quantities scales with length, we can redefine the variables as
\begin{gather*}
g\rightarrow\tau^2g, \quad\quad\k\rightarrow\tau^{-1}\k, \quad\quad N\rightarrow\tau^2N, \quad\quad X\rightarrow\tau^{-1}X, \\
A_i\rightarrow A_i, \quad\quad \E\rightarrow\tau\E, \quad\quad \phi\rightarrow\tau^{-1}\phi,
\end{gather*}

\noindent where the new $(\k,g,\E,A,N,X,\phi)$ are all dimensionless quantities. We will finally remove the explicit $\tau$ dependence by defining the dimensionless time coordinate $T$ by
\begin{equation*}
T=-\ln\bigg(\frac{\tau}{\tau_0}\bigg),
\end{equation*}

\noindent where $\tau_0\in(-\infty,0)$ is some arbitrary reference time. This allows us to write in the rescaled Einstein's equations $\tau\partial_\tau=-\partial_T$ and $\tau=\tau_0e^{-T}$. We also note that, because $\tau\in(-\infty,0)$, we have that $T\in\mathbb{R}$ and in particular as we evolve forward in time such that $\tau\rightarrow0$, then $T\rightarrow\infty$.

Substituting the dimensionless variables $(\k,g,\E,A,N,X,\phi,T)$ into the evolution and constraint equations given in eqns. (46)--(51) of \cite{Mondal2021}, we find the system of scale-invariant evolution and constraint equations to be as follows. The scale-invariant evolution equations are given by
\begin{align}
\partial_TA^a{}_{bi}&=-N\E^a{}_{bi}-\partial_i\phi^a{}_b-[A_i,\phi]^a{}_b-X^j\nabla[\gamma]_jA^a{}_{bi}-A^a{}_{bj}\nabla[\gamma]_iX^j \label{A-evol} \\
\partial_T\E^a{}_{bi}&=-\bigg(\frac{(n-2)N}{n}-1\bigg)\E^a{}_{bi}+(\nabla_jN)F^a{}_{bi}{}^j+N(\nabla_jF^a{}_{bi}{}^j-\nabla_iC^a{}_b) \notag\\
&\quad-N[F_i{}^j,A_j]^a{}_b-[\E_i,\phi]^a{}_b-X^k\nabla_k\E^a{}_{bi}-\E^a{}_{bk}\nabla_iX^k+2N\k_i^j\E^a{}_{bj} \label{E-evol} \\
\partial_Tg_{ij}&=-\frac{2}{n}(n-N)g_{ij}+2N\k_{ij}-L_Xg_{ij} \label{g-evol} \\
\partial_T\k_{ij}&=-\bigg(\frac{(n-2)N}{n}+1\bigg)\k_{ij}-\frac{N-1}{n}g_{ij}+\nabla_i\nabla_jN-N\big\{R_{ij}-\alpha_{ij}-2\k_{ik}\k_j{}^k\big\} \notag\\
&\quad-\tau_0^2e^{-2T}N\bigg\{\E^a{}_{bi}\E^b{}_{aj}-g^{k\ell}F^a{}_{bik}F^b{}_{aj\ell}-\frac{1}{n-1}g^{k\ell}\E^a{}_{bk}\E^b{}_{a\ell}g_{ij} \notag\\
&\quad+\frac{1}{2(n-1)}g^{km}g^{\ell n}F^a{}_{bk\ell}F^b{}_{amn}g_{ij}\bigg\}-L_X\k_{ij}. \label{k-evol}
\end{align}

\noindent Here $\alpha_{ij}$ and $C^a{}_b$ has been defined as
\begin{align*}
\alpha_{ij}&=\frac{1}{2}(\nabla_iV_j+\nabla_jV_i) \\
C^a{}_b&=g^{ij}\nabla^{\widehat{A}}[\gamma]_i(A^a{}_{bj}-\widehat{A}^a{}_{bj})
\end{align*}

\noindent and must vanish for the evolution equations to be hyperbolic, as is the case in the chosen gauges; c.f. Section 3 of \cite{Mondal2021}. We note that the first term in \eqref{E-evol} will, to leading order in the small data context we work in, vanish for $n=3$. This ultimately will prevent the subsequent energy arguments from closing. While this is a feature of the gauge and scaling choices made earlier, in section \ref{gauge-covariant}, we give an gauge-covariant argument as to why an energy estimate will not be sufficient to obtain the desired global existence result in $n=3$ spatial dimensions.

Meanwhile, the elliptic equations for the gauge variables take the form
\begin{align}
\Delta_gN&+\bigg(|\k|^2+\frac{1}{n}+\frac{n-2}{n-1}\tau_0^2e^{-2T}g^{k\ell}\E^a{}_{bk}\E^b{}_{a\ell}+\frac{\tau_0^2}{2(n-1)}e^{-2T}g^{k\ell}g^{mn}F^a{}_{bkm}F^b{}_{a\ell n}\bigg)N=1 \label{N-ell} \\
\Delta_gX^i&-R^i{}_jX^j+L_XV^i=-2\k^{ij}\nabla_jN+\bigg(1-\frac{2}{n}\bigg)\nabla^iN-2\tau_0^2e^{-2T}Ng^{ik}g^{\ell j}F^a{}_{bk\ell}\E^b{}_{aj} \notag\\
&\qquad\qquad\qquad\qquad\quad\,+\bigg(2N\k^{jk}+\frac{2N}{n}g^{jk}-2\nabla^jX^k\bigg)\big(\Gamma[g]^i{}_{jk}-\Gamma[\gamma]^i{}_{jk}\big)+g^{jk}\partial_T\Gamma[\gamma]^i{}_{jk} \label{X-constraint} \\
&\!\!\!\!\!\!\!\!\!\!\!\!\!\!\!g^{ij}\nabla[\gamma]^A_i\nabla[\gamma]^A_j\phi^a{}_b+g^{ij}[\widehat{A}_i-A_i,\nabla^A_j\phi]^a{}_b+g^{ij}\nabla_iN\E^a{}_{bj}-Ng^{ij}[A_i,\E_j]^a{}_b \notag\\
&\quad+g^{ij}[\widehat{A}_i,N\E_j]^a{}_b+g^{ij}\big(\Gamma[g]^k{}_{ij}-\Gamma[\gamma]^k{}_{ij}\big)N\E^a{}_{bk}-X^k\nabla[\gamma]_kg^{ij}\nabla[\gamma]_i\big(A^a{}_{bj}-\widehat{A}^a{}_{bj}\big) \notag\\
&\quad-X^k\nabla[\gamma]_kg^{ij}[\widehat{A}_i,A_j-\widehat{A}_j]^a{}_b+g^{ij}\nabla[\gamma]_iX^k\nabla[\gamma]_kA^a{}_{bj}+g^{ij}\nabla[\gamma]_iA^a{}_{bk}\nabla[\gamma]_jX^k \notag \\
&\quad+g^{ij}A^a{}_{bk}\nabla[\gamma]_i\nabla[\gamma]_jX^k-g^{ij}R[\gamma]^\ell{}_{jik}A^a{}_{b\ell}X^k+g^{ij}[\widehat{A}_i,A_k]^a{}_b\nabla[\gamma]_jX^k \notag \\
&\quad+\bigg(2N\k^{ij}+\frac{2N}{n}g^{ij}-\nabla^iX^j-\nabla^jX^i\bigg)\nabla[\gamma]^{\widehat{A}}_i(A^a{}_{bj}-\widehat{A}^a{}_{bj})=0 \label{phi-constraint}
\end{align}

\noindent Throughout the paper, we will collectively refer to the rescaled, gauge-fixed equations \eqref{A-evol}--\eqref{phi-constraint} as the Einstein-Yang-Mills system. Due to the elliptic contraints, if $\mathcal{D}_0$ denotes the space of diffeomorphism on $\Sigma$ and $\mathcal{G}$ denotes the space of gauge transformations, then the dynamics of the Einstein-Yang-Mills system will occur on the space $\mathcal{M}_\Sigma/\mathcal{D}_0\times\mathcal{A}/\mathcal{G}$.

Note that indices are raised with the dynamical metric $g$, unless otherwise specified. Recall that because we have no equivalent shadow metric condition on $\widehat{A}$, we have fixed $\widehat{A}\equiv0$. However, because $A-\widehat{A}$ is a gauge invariant tensor while $A$ itself is not, we for now explicitly write all $\widehat{A}$ terms in the constraint equation for $\phi$.

Throughout the analysis in this paper, we will assume that $(\k,g,\E,A)\in H^{s-1}\times H^s\times H^{s-1}\times H^s$ for $s>\frac{n}{2}+1$. From the elliptic equations, we then immediately see that the gauge variables $(N,X,\phi)$ each lie in the Sobolev space $H^{s+1}$.

\subsection{Small Data}
We work in the context of a background geometry given by a metric of the form in \eqref{M-background-metric}, which is a solution to the vacuum Einstein's equations. From this unscaled metric, we then have the background solution $(\k,g,\E,A,N,X,\phi)=(0,t^{-2}\gamma,0,0,t^{-2}n,0,0)$. Introducing the rescaled variables, this corresponds to the background solution
\begin{equation}\label{background-sol}
(\k,g,\E,A,N,X,\phi)_\text{background}=(0,\gamma,0,0,n,0,0).
\end{equation}

\noindent If we define $u\coloneqq g-\gamma$ and $\omega\coloneqq N-n$, we then will consider the data
\begin{equation*}
(\k,u,\E,A,\omega,X,\phi)\in H^{s-1}\times H^s\times H^{s-1}\times H^s\times H^{s+1}\times H^{s+1}\times H^{s+1}
\end{equation*}

\noindent to be small in the small data context we are working in, for $s>\frac{n}{2}+1$. We will often specify this by writing that $(\k,u,\E,A,\omega,X,\phi)\in B_\delta(0)$, where $B_\delta(0)$ is a ball of sufficiently small radius $\delta$ centered at $0$ in the relevant function spaces.

\section{Elliptic Estimates}\label{elliptic}
To begin, we will recall a number of important inequalities regarding Sobolev spaces, which we will make frequent use of in the following elliptic and energy estimates.
\begin{proposition}\label{important-ineqs}
Let $OP^r$ denote the space of pseudo-differential operators with symbol in H\"omander's class $S^r_{1,0}$. The following inequalities then hold.

\begin{enumerate}[label=\textbf{\arabic*}.]
\item \cite{Taylor1991} If $\mathcal{P}\in OP^r$ for $r\in\mathbb{R}$ and if $s\in\mathbb{R}$, then $\mathcal{P}:H^s\rightarrow H^{s-r}$. That is, if $f\in H^s$, then the estimate
\begin{equation}
\hnorm{\mathcal{P}f}{s-r}\leq C\hnorm{f}{s}
\end{equation}

\noindent holds.

\item \cite{Taylor1991} Let $\mathcal{P}\in OP^r$ for $r>0$ and take $s\geq0$. Then for $f\in H^{r+s}$ and $g\in H^{r+s-1}$,
\begin{equation}
\hnorm{[\mathcal{P},f]g}{s}=\hnorm{\mathcal{P}(fg)-f\mathcal{P}(g)}{s}\leq C\big(\lnorm{\nabla[\gamma]f}{\infty}\hnorm{g}{r+s-1}+\hnorm{f}{r+s}\lnorm{g}{\infty}\big)
\end{equation}

\item \cite{Kato1988} If $s>0$ then $H^s\cap L^\infty$ is an algebra and for $f,g\in H^s\cap L^\infty$,
\begin{equation}
\hnorm{fg}{s}\leq C\big(\lnorm{f}{\infty}\lnorm{g}{2}+\lnorm{f}{2}\lnorm{f}{\infty}\big).
\end{equation}

\noindent In particular, with the embeddings $H^s\hookrightarrow L^2$ for all $s$ and $H^s\hookrightarrow W^{0,\infty}=L^\infty$ for $s>\frac{n}{2}$, then if $s>\frac{n}{2}$,
\begin{equation}
\hnorm{fg}{s}\leq C\hnorm{f}{s}\hnorm{g}{s}.
\end{equation}
\end{enumerate}
\end{proposition}

\noindent An important case of Proposition \ref{important-ineqs}.1 is $\nabla[\gamma]\in OP^1$, for which we have the estimate $\hnorm{\nabla[\gamma]f}{s-1}\leq C\hnorm{f}{s}$. Moreover, we will make use of---although not always explicitly---Minkowski's and H\"older's inequalities when necessary.

Now we will obtain estimates for the small gauge variables $\omega$, $X$, and $\phi$. This, however, will first require an estimate on the tangential velocity $\partial_T\gamma$, given in the following lemma.
\begin{lemma}\label{gamma-evol}
Let $\gamma\in\mathcal{N}$ be a metric satisfying the shadow metric condition $(g-\gamma)\perp\mathcal{N}$ and for which $\hnorm{g-\gamma}{s}<\delta$ with a sufficiently small $\delta>0$ and $s>\frac{n}{2}+1$. Then the vector field $\partial_T\gamma$ satisfies the estimate
\begin{equation}
\norm{\partial_T\gamma}\leq C\big\{(1+\hnorm{u}{s})(\hnorm{\omega}{s+1}+\hnorm{X}{s+1})+\hnorm{\k}{s-1}(1+\hnorm{\omega}{s+1})\big\},
\end{equation}

\noindent for a constant $C$ depending only on the background geometry.
\end{lemma}

\begin{proof}
Recall \eqref{dT-gamma-est}; that is, from the shadow metric condition, we obtain an estimate for the time evolution $\partial_T\gamma$ of the form
\begin{equation*}
\norm{\partial_T\gamma}\leq C\hnorm{\partial_Tg}{s-1}.
\end{equation*}

\noindent In the evolution equation for $g$, we write the Lie derivative $\mathcal{L}_Xg$ in terms of $u=g-\gamma$ and $\gamma$ as
\begin{align*}
\mathcal{L}_Xg_{ij}=\mathcal{L}_Xu_{ij}+\mathcal{L}_X\gamma_{ij}=X^k\nabla[\gamma]_ku_{ij}+u_{ik}\nabla[\gamma]_jX^k+u_{jk}\nabla[\gamma]_iX^k+\nabla[\gamma]_iX_j+\nabla[\gamma]_jX_i.
\end{align*}

\noindent We then have from the evolution equation for $g$ and Minkowski's inequality, that $\partial_T\gamma$ satisfies the norm bound
\begin{equation*}
\norm{\partial_T\gamma}\leq C\big\{\hnorm{u\omega}{s-1}+\hnorm{\omega}{s-1}+\hnorm{\k}{s-1}+\hnorm{\k\omega}{s-1}+\hnorm{X\nabla[\gamma]u}{s-1}+\hnorm{u\nabla[\gamma]X}{s-1}+\hnorm{\nabla[\gamma]X}{s-1}\big\},
\end{equation*}

\noindent where the positive constant $C=C(n,\norm{\gamma},\dots)$ depends on the background geometry only. Applying the inequalities from Proposition \ref{important-ineqs} for $s-1>\frac{n}{2}$, we obtain
\begin{align*}
\norm{\partial_T\gamma}&\leq C\big\{\hnorm{\omega}{s-1}\hnorm{u}{s-1}+\hnorm{\omega}{s-1}+\hnorm{\k}{s-1}+\hnorm{\k}{s-1}\hnorm{\omega}{s-1} \\
&\qquad\qquad+\hnorm{X}{s-1}\hnorm{\nabla[\gamma]u}{s-1}+\hnorm{\nabla[\gamma]X}{s-1}\hnorm{u}{s-1}+\hnorm{\nabla[\gamma]X}{s-1}\big\} \\
&\leq C\big\{\hnorm{\omega}{s-1}\hnorm{u}{s-1}+\hnorm{\omega}{s-1}+\hnorm{\k}{s-1}+\hnorm{\k}{s-1}\hnorm{\omega}{s-1} \\
&\qquad\qquad+\hnorm{X}{s-1}\hnorm{u}{s}+\hnorm{X}{s}\hnorm{u}{s-1}+\hnorm{X}{s}\big\} \\
\end{align*}

\noindent Using the embedding $H^{s_1}\hookrightarrow H^{s_2}$ for $s_1>s_2$, we obtain the desired estimate.
\end{proof}

We now study the perturbed lapse function $\omega=N-n$, where we will use the proof of a similar result from \cite{Mondal2022} to show that $\omega$ satisfies the following second-order estimate.
\begin{lemma}\label{omega-estimate}
Fix $s>\frac{n}{2}+1$. Let $B_\delta(0)$ be a ball of sufficiently small radius $\delta$ containing $(\k,u,\E,A)$ that satisfy the Einstein-Yang-Mills system. Then the lapse gauge variable $\omega=N-n$ satisfies the estimate
\begin{equation}\label{omega-ineq}
\hnorm{\omega}{s+1}\leq C\Big\{\hnorm{\k}{s-1}^2+e^{-2T}\big(\hnorm{\E}{s-1}^2+\hnorm{A}{s}^2\big)\Big\},
\end{equation}

\noindent for a constant $C$ depending only on the background geometry.
\end{lemma}

\begin{proof}
From Eq. \eqref{N-ell}, the elliptic equation for the perturbation $\omega$ will be given by
\begin{align}\label{P-omega}
\Delta_g\omega+\bigg(|\k|^2+\frac{1}{n}+\frac{n-2}{n-1}\tau_0^2e^{-2T}|\E|^2&+\frac{\tau_0^2}{2(n-1)}e^{-2T}|F|^2\bigg)\omega \notag \\
&=-n\bigg(|\k|^2+\frac{n-2}{n-1}\tau_0^2e^{-2T}|\E|^2+\frac{\tau_0^2}{2(n-1)}e^{-2T}|F|^2\bigg).
\end{align}

\noindent We note that the second term of the left-hand side is strictly positive, and thus the second-order elliptic operator $\mathcal{P}=\Delta_g+\big(|\k|^2+\frac{1}{n}+\frac{n-2}{n-1}\tau_0^2e^{-2T}|\E|^2+\frac{\tau_0^2}{2(n-1)}e^{-2T}|F|^2\big)\Id$ has a definite sign. As shown via contradiction in the proof of Lemma 5 in \cite{Mondal2022}, such an operator will satisfy an estimate of the form $\hnorm{\omega}{s+1}\leq C\hnorm{\mathcal{P}\omega}{s-1}$ for any $s\geq1$. Since we have assumed throughout that $s>\frac{n}{2}+1>1$, we thus have such an estimate with $\mathcal{P}\omega$. Using the expression for $\mathcal{P}\omega$ in \eqref{P-omega}, we then may bound the Sobolev norm of $\omega$ as
\begin{align*}
\hnorm{\omega}{s+1}&\leq C\bigg\lVert-n\bigg(|\k|^2+\frac{n-2}{n-1}\tau_0^2e^{-2T}|\E|^2+\frac{\tau_0^2}{2(n-1)}e^{-2T}|F|^2\bigg)\bigg\rVert_{H^{s-1}} \\
&\leq C\big\{\hnorm{\k}{s-1}^2+e^{-2T}\hnorm{\E}{s-1}^2+e^{-2T}\hnorm{F}{s-1}^2\big\} \\
&\leq C\big\{\hnorm{\k}{s-1}^2+e^{-2T}\hnorm{\E}{s-1}^2+e^{-2T}\hnorm{A}{s}^2(1+\hnorm{A}{s})^2\big\}
\end{align*}

\noindent Because we have taken $(\k,A,\E)\in B_\delta(0)$, we may absorb non-leading order terms into the coefficient $C$ by redefining $C\rightarrow(1+\delta)C$, i.e.,
\begin{align*}
\hnorm{\omega}{s+1}&\leq C\hnorm{\k}{s-1}^2+Ce^{-2T}\hnorm{\E}{s-1}^2+Ce^{-2T}\hnorm{A}{s}^2(1+\delta)^2 \\
&\leq C\hnorm{\k}{s-1}^2+Ce^{-2T}\hnorm{\E}{s-1}^2+C'e^{-2T}\hnorm{A}{s}^2 \\
&\leq C'\big\{\hnorm{\k}{s-1}^2+e^{-2T}(\hnorm{\E}{s-1}^2+\hnorm{A}{s}^2)\big\}.
\end{align*}
\end{proof}

\noindent We note that, as we will study the Cauchy problem and thus take initial data at some $T_0$, we may bound $e^{-T}\leq e^{-T_0}$ which may be absorbed into the constant $C$. However, if terms of the form $e^T$ were to appear, our subsequent energy estimates would fail as these exponential terms would prevent us from finding a decreasing bound on the total energy. Since $e^{-T}$ terms may neutralize such terms, we will keep, for now, all terms of the form $e^{-T}$ with their leading order Sobolev norms.

While the elliptic differential operator for $\omega$ has definite sign, those appearing in the defining equations for $X$ and $\phi$ do not. This means that, in general, the elliptic operators acting on $X$ and $\phi$ will not have trivial kernel and thus will exhibit a \textit{Gribov ambiguity} \cite{Moncrief1979,Singer1978}. Geometrically we can view this degeneracy as a breakdown of geodesic normal coordinates on the gauge-fixed spaces of Riemannian metrics or connection $1$-forms, as follows.

Let us consider the shift vector field $X$. After fixing the spatial harmonic gauge, the dynamics of the metric evolution occur on the space $\mathcal{M}_\Sigma/\mathcal{D}_0$, where we recall that $\mathcal{D}_0$ is the space of diffeomorphisms on $\Sigma$. However, the elliptic equation defining $X$ is derived from requiring that the spatial harmonic gauge condition is preserved over time, i.e., $\partial_TV^k=0$ for the tension field $V^k$. Since geodesics on the gauge-fixed quotient space $\mathcal{M}_\Sigma/\mathcal{D}_0$ must preserve the spatial harmonic gauge condition, we then have that the uniqueness of geodesics will be determined by the uniqueness of the $X$ defined by \eqref{X-constraint}. In other words, looking at geodesic normal coordinates, we will have a neighborhood $B_\delta(\gamma)$ of a metric $\gamma\in\mathcal{M}_\Sigma/\mathcal{D}_0$ such that the exponential map used to construct geodesics is injective, which in turn corresponds to the injectivity of the differential operator in \eqref{X-constraint} that acts on $X$. However, for $g\in(\mathcal{M}_\Sigma/\mathcal{D}_0)\setminus B_\delta(\gamma)$, there will not be a unique geodesic between $\gamma$ and $g$, corresponding to the aforementioned Gribov ambiguity and breakdown in injectivity of the differential operator defining $X$.

An identical reasoning ought to hold for $\phi$ with the space of connection $1$-forms $A$, though as the topology of this space is unknown we can only draw parallels. We see from this geometric view that we will want to restrict our focus to neighborhoods in the spaces of Riemannian metrics and connection $1$-forms on which the elliptic operators defining $X$ and $\phi$ remain injective, which we may freely do as we are working in the context of small data. The following lemma gives the necessary bound for the shift vector field $X$. We give the full proof, following the derivation given in \cite{Andersson2003} and \cite{Mondal2022}, as we will use a similar method to subsequently estimate $\phi$.

\begin{lemma}\label{X-estimate}
Fix $s>\frac{n}{2}+1$. Let $B_\delta(0)$ be a ball of sufficiently small radius $\delta$ containing $(\k,u,\E,A)$ that satisfy the Einstein-Yang-Mills system. Then the shift vector field $X$ will satisfy the bound
\begin{equation}
\hnorm{X}{s+1}\leq C\big\{\hnorm{\k}{s-1}+\hnorm{u}{s}+e^{-2T}(\hnorm{\E}{s-1}^2+\hnorm{A}{s}^2)+e^{-4T}(\hnorm{\E}{s-1}^2+\hnorm{A}{s}^2)^2\big\},
\end{equation}

\noindent for a constant $C$ depending only on the background geometry.
\end{lemma}

\begin{proof}
Recall the elliptic equation for $X$ is given as
\begin{align*}
\Delta_gX^i-R^i{}_jX^j+L_XV^i&=-2\k^{ij}\nabla_jN+\bigg(1-\frac{2}{n}\bigg)\nabla^iN-2\tau_0^2e^{-2T}Ng^{ik}g^{\ell j}F^a{}_{bk\ell}\E^b{}_{aj} \notag\\
&\qquad+\big(2N\k^{jk}-2\nabla^jX^k\big)\big(\Gamma[g]^i{}_{jk}-\Gamma[\gamma]^i{}_{jk}\big)+g^{jk}\partial_T\Gamma[\gamma]^i{}_{jk}.
\end{align*}

\noindent It is necessary to then demonstrate that $X^i\mapsto\Delta_gX^i+2\nabla^jX^k(\Gamma[g]^i{}_{jk}-\Gamma[\gamma]^i{}_{jk})-R^i{}_jX^j+L_XV^i$ is an isomorphism of function spaces $H^{s+1}$ and $H^{s-1}$. To do so, we will consider the flow of $X$ along the pull-back of the tension field which, in the CMCSH gauge, will vanish.

Let $\psi_s$ be the flow of $X$, and consider the pull-back of the tension field $V^\flat$. As the difference of connections transforms as a section of a suitable vector bundle, this gives
\begin{equation*}
\big(\psi_s^*V^\flat\big)^i=\big(\psi_s^*g^{-1}\big)^{jk}\big(\psi_s^*(\Gamma[g]-\Gamma[\gamma])\big)^i{}_{jk}=\big(\psi_s^*g^{-1}\big)^{jk}\big(\Gamma[\psi_s^*g]^i{}_{jk}-\Gamma[\psi_s^*\gamma]^i{}_{jk}\big).
\end{equation*}

\noindent We may then compute the Lie derivative of $V^\flat$ along $X$ as
\begin{align*}
(L_XV)^i&=\der{}{s}(\psi_s^*V^\flat)^i\big|_{s=0} \\
&=(L_Xg^{-1})^{jk}\big(\Gamma[g]^i{}_{jk}-\Gamma[\gamma]^i{}_{jk}\big)+\frac{1}{2}g^{jk}g^{i\ell}\big(-\nabla_\ell(L_Xg)_{jk}+\nabla_j(L_Xg)_{\ell k}+\nabla_k(L_Xg)_{\ell j}\big) \\
&\qquad-\frac{1}{2}g^{jk}\gamma^{i\ell}\big(-\nabla[\gamma]_\ell(L_X\gamma)_{jk}+\nabla[\gamma]_j(L_X\gamma)_{\ell k}+\nabla[\gamma]_k(L_X\gamma)_{\ell j}\big)
\end{align*}

\noindent Expanding the Lie derivative on the right-hand side with $(L_Xg)_{ij}=\nabla_iX_j+\nabla_jX_i$ and $(L_X\gamma)_{ij}=\nabla[\gamma]_iX_j+\nabla[\gamma]_jX_i$, we then find that
\begin{equation}\label{X-op-simplification}
\Delta_gX^i-R^i{}_jX^j+L_XV^i+2\nabla^jX^k\big(\Gamma[g]^i{}_{jk}-\Gamma[\gamma]^i{}_{jk}\big)=\Delta^\gamma_gX^i-g^{k\ell}R[\gamma]^i{}_{kj\ell}X^j,
\end{equation}

\noindent The operator acting on $X$ on the left-hand side is $\mathcal{P}$, and the operator on the right-hand side now has a definite sign. So to show that $\mathcal{P}$ is an isomorphism of function spaces is to show that the mapping $\mathcal{P}_{g,\gamma}:X^i\mapsto\Delta_\gamma^gX^i-g^{k\ell}R[\gamma]^i{}_{kj\ell}X^j$ is. As the Cauchy hypersurface is compact and $\mathcal{P}_{g,\gamma}$ is second-order elliptic, then it will suffice to show that $\mathcal{P}_{g,\gamma}$ is injective in order to show it is the desired isomorphism.

Suppose, for the sake of contradiction, that $\mathcal{P}_{g,\gamma}$ is not injective. We thus have a non-zero $Z\in\ker\mathcal{P}_{g,\gamma}$. We may the integrate $\gamma_{ij}Z^i\mathcal{P}_{g,\gamma}Z^j\equiv0$ over the Cauchy hypersurface $\Sigma$, yielding
\begin{align*}
0&=\int_\Sigma\big(-\gamma_{ij}g^{k\ell}Z^j\nabla[\gamma]_k\nabla[\gamma]_\ell Z^i-\gamma_{im}g^{k\ell}Z^mR[\gamma]^i{}_{kj\ell}Z^j\big)\:\mu_\gamma \\
&=\int_\Sigma\big(\gamma_{ij}g^{k\ell}\nabla[\gamma]_kZ^j\nabla[\gamma]_\ell Z^i-\gamma_{im}g^{k\ell}Z^mR[\gamma]^i{}_{kj\ell}Z^j\big)\:\mu_\gamma \\
&=\int_\Sigma\big(\gamma_{ij}\gamma^{k\ell}\nabla[\gamma]_kZ^j\nabla[\gamma]_\ell Z^i+\gamma_{ij}(g^{k\ell}-\gamma^{k\ell})\nabla[\gamma]_kZ^j\nabla[\gamma]_\ell Z^i \\
&\qquad\qquad-\gamma_{im}\gamma^{k\ell}Z^mR[\gamma]^i{}_{kj\ell}Z^j-\gamma_{im}(g^{k\ell}-\gamma^{k\ell})Z^mR[\gamma]^i{}_{kj\ell}Z^j\big)\:\mu_\gamma.
\end{align*}

\noindent With $\hnorm{u}{s}$ taken to be sufficiently small, we have that the second and fourth terms in the integrand are domainated by the first and third terms. Furthermore, because $\gamma^{k\ell}R[\gamma]^i{}_{kj\ell}=R[\gamma]^i{}_j$ is negative-definite, we have that the first and third terms of the integrand are positive. Hence the integrand as a whole may be taken to be non-negative, implying that $Z$ must be $0$. This yields a contradiction, and so we must have that $\ker\mathcal{P}_{g,\gamma}=\{0\}$. Hence $\mathcal{P}_{g,\gamma}$ is injective, and because $\mathcal{P}=\mathcal{P}_{g,\gamma}$ by \eqref{X-op-simplification}, we have that $\mathcal{P}$ is injective and thus an isomorphism of suitable function spaces. By elliptic regularity, we have an estimate of the form $\hnorm{X}{s+1}\leq\hnorm{\mathcal{P}X}{s-1}$. Using the result of Lemma \ref{gamma-evol}, it then follows that $X$ will be controlled by
\begin{align*}
\hnorm{X}{s+1}\big\{1{}&{}-C(\hnorm{\gamma-\gamma_0}{s}+\hnorm{u}{s})(1+\hnorm{u}{s})\big\}\leq C(1+\hnorm{\k}{s-1})\hnorm{\omega}{s+1} \notag\\
&+C(1+\hnorm{\omega}{s+1})\big\{\hnorm{\k}{s-1}\hnorm{u}{s}+\tau_0^2e^{-2T}\hnorm{F}{s-1}\hnorm{\E}{s-1}\big\} \\
&+C(\hnorm{\gamma-\gamma_0}{s}+\hnorm{u}{s})\big\{(1+\hnorm{u}{s})(\hnorm{u}{s}+\hnorm{\omega}{s+1})+(1+\hnorm{\omega}{s+1})\hnorm{\k}{s-1}\big\}.
\end{align*}

\noindent Note that $\hnorm{\gamma-\gamma_0}{s}\leq C$ for constant $C$ depending only on the background geometry. Hence with sufficiently small $\hnorm{u}{s}$, and using the estimate for $\omega$ from Lemma \ref{omega-estimate}, the shift vector field $X$ is controlled to leading order by
\begin{align*}
\hnorm{X}{s+1}&\leq C\big\{\hnorm{\k}{s-1}+\hnorm{u}{s}+e^{-2T}(\hnorm{\E}{s-1}^2+\hnorm{A}{s}^2)+e^{-4T}(\hnorm{\E}{s-1}^2+\hnorm{A}{s}^2)^2\big\}.
\end{align*}
\end{proof}

\noindent We now give a bound on the norm of the Yang-Mills gauge variable $\phi$. As previously mentioned, the proof will be similar in structure to that of Lemma \ref{X-estimate}: we will consider the pull-back of the generalized Coulomb gauge condition along the flow of $\phi$ and compute its derivative at $s=0$ to find a form of the elliptic operator of interest that can then be shown to be an isomorphism of function spaces.

\begin{lemma}\label{phi-estimate}
Fix $\widehat{A}=0$ and $s>\frac{n}{2}+1$. Let $B_\delta(0)$ be a ball of sufficiently small radius $\delta$ containing $(\k,u,\E,A)$ that satisfy the Einstein-Yang-Mills system. Then the Yang-Mills gauge variable $\phi$ will satisfy the estimate
if $\widehat{A}=0$:
\begin{align}
\hnorm{\phi}{s+1}&\leq C\big\{\hnorm{\k}{s-1}^2+\hnorm{u}{s}^2+\hnorm{\E}{s-1}^2+\hnorm{A}{s}^2+e^{-2T}(\hnorm{\E}{s-1}+\hnorm{A}{s})(\hnorm{\E}{s-1}^2+\hnorm{A}{s}^2) \notag \\
&\qquad\quad+e^{-4T}\hnorm{A}{s}(\hnorm{\E}{s-1}^2+\hnorm{A}{s}^2)^2\big\},
\end{align}

\noindent for a constant $C$ depending only on the background geometry.
\end{lemma}

\begin{proof}
To begin, we will show that the mapping $\mathcal{P}:\phi\mapsto g^{ij}\nabla[\gamma]_i^A\nabla[\gamma]_j^A\phi+g^{ij}[\widehat{A}_i-A_i,\nabla_j^A\phi]$ is an isomorphism of suitable function spaces. Let $\psi_s=e^{s\phi}$ be the one-parameter subgroup of the Lie group $G$ corresponding to $\phi$ and take $\mathcal{Q}=g^{ij}\nabla[\gamma]_i^{\widehat{A}}(A_j-\widehat{A}_j)$.

As the difference of two connections will transform under a gauge transformation as a suitable section of the associated $V$-bundle, we have that under a gauge transformation $\mathcal{G}$, $\mathcal{Q}$ will transform as $\mathcal{Q}\mapsto\mathcal{GQG}^{-1}$. Taking $\mathcal{G}$ to be $\psi_s$ and differentiating with respect to $s$,
\begin{align*}
\der{}{s}\big(e^{s\phi}\mathcal{Q}e^{-s\phi}\big)_{s=0}&=\der{}{s}\Big(g^{ij}\nabla[\gamma]_ie^{s\phi}(A_j-\widehat{A}_j)e^{-s\phi}+g^{ij}[e^{s\phi}\widehat{A}_ie^{-s\phi},e^{s\phi}(A_j-\widehat{A}_j)e^{-s\phi}]\Big)_{s=0} \notag\\
[\phi,\mathcal{Q}]&=g^{ij}\nabla[\gamma]_i[\phi,A_j-\widehat{A}_j]+g^{ij}\big[[\phi,\widehat{A}_i],A_j-\widehat{A}_j\big]+g^{ij}\big[\widehat{A}_i,[\phi,A_j-\widehat{A}_j]\big].
\end{align*}

\noindent An explicit computation gives
\begin{equation}
[\phi,\mathcal{Q}]=-\mathcal{P}\phi+g^{ij}\Big\{\nabla[\gamma]_i^{\widehat{A}}\nabla[\gamma]_j^{\widehat{A}}\phi+\big[A_i-\widehat{A}_i,[\widehat{A}_j,\phi]\big]\Big\},
\end{equation}

\noindent and so in the generalized Coulomb gauge, where $\mathcal{Q}=0$,
\begin{equation}
\mathcal{P}\phi^a{}_b=g^{ij}\nabla[\gamma]_i^{\widehat{A}}\nabla[\gamma]_j^{\widehat{A}}\phi^a{}_b+g^{ij}\big[A_i-\widehat{A}_i,[\widehat{A}_j,\phi]\big]^a{}_b.
\end{equation}

\noindent To see that $\mathcal{P}$ is an isomorphism, suppose that there is a non-zero $\zeta\in\ker\mathcal{P}$. Integrating the gauge-invariant inner product $\zeta^b{}_a\mathcal{P}\zeta^a{}_b\equiv0$ over $\Sigma$,
\begin{align*}
0&=\int_\Sigma\Big(g^{ij}\zeta^b{}_a\nabla[\gamma]_i^{\widehat{A}}\nabla[\gamma]_j^{\widehat{A}}\zeta^a{}_b+g^{ij}\zeta^b{}_a\big[A_i-\widehat{A}_i,[\widehat{A}_j,\zeta]\big]^a{}_b\Big)\:\mu_\gamma \\
&=\int_\Sigma\Big(-\gamma^{ij}\nabla[\gamma]_i^{\widehat{A}}\zeta^b{}_a\nabla[\gamma]_j^{\widehat{A}}\zeta^a{}_b-\big(g^{ij}-\gamma^{ij}\big)\nabla[\gamma]_i^{\widehat{A}}\zeta^b{}_a\nabla[\gamma]_j^{\widehat{A}}\zeta^a{}_b \\
&\qquad\qquad+\gamma^{ij}\zeta^b{}_a\big[A_i-\widehat{A}_i,[\widehat{A}_j,\zeta]\big]^a{}_b+\big(g^{ij}-\gamma^{ij}\big)\zeta^b{}_a\big[A_i-\widehat{A}_i,[\widehat{A}_j,\zeta]\big]^a{}_b\Big)\:\mu_\gamma
\end{align*}

\noindent Now, with $\widehat{A}=0$, the third and fourth terms will drop out. Meanwhile, with $\hnorm{u}{s}$ small, the second term can be taken to be smaller in magnitude than the first. Since the first term is negative-definite, we have that the integrand as a whole will be non-positive. As such, we must have $\zeta=0$ for the integral to vanish.

We conclude that there can be no non-zero $\zeta\in\ker\mathcal{P}$, and so since $\mathcal{P}$ has trivial kernel it must be injective. As $\Sigma$ is compact and $\mathcal{P}$ is second-order elliptic, the injectivity of $\mathcal{P}$ gives that it is indeed an isomorphism of $H^{s+1}$ and $H^{s-1}$. So, having established that $\mathcal{P}:\phi\mapsto g^{ij}\nabla[\gamma]_i^A\nabla[\gamma]_j^A\phi+g^{ij}[\widehat{A}_i-A_i,\nabla_j^A\phi]$ is an isomorphism between function spaces, it then follows by elliptic regularity and the inequalities outlined in Proposition \ref{important-ineqs} that
\begin{align}
\hnorm{\phi}{s+1}&\leq C\big\{\hnorm{\E}{s-1}\hnorm{\omega}{s+1}+\hnorm{A}{s}(1+\hnorm{\omega}{s+1})(\hnorm{\E}{s-1}+\hnorm{\k}{s-1}) \notag \\
&\qquad\qquad\qquad\qquad\qquad\qquad\qquad\qquad\quad+(1+\hnorm{u}{s})\hnorm{A}{s}\hnorm{X}{s+1}\big\}.
\end{align}

\noindent Using the estimates for the lapse function perturbation $\omega$ and shift field $X$ from Lemmas \ref{omega-estimate} and \ref{X-estimate}, the desired estimate for $\phi$ is then obtained.
\end{proof}

Finally, because the deformation space $\mathcal{N}$ is not in general trivial, recall that we may decompose a tangent vector field to the shadow metric as $\partial_T\gamma=h^{\TT\parallel}+L_{Y^\parallel}\gamma$, with $Y^\parallel$ in general non-zero. It will turn out that we will require an estimate of the sum $X+Y^\parallel$ to complete the energy estimates in the gravitational sector, due to a term of the form $\mathcal{L}_{X+Y^\parallel}\gamma$. We present a bound on the Sobolev norm for $X+Y^\parallel$ below, using the result of \cite{Andersson2009} that a particular second-order differential operator is indeed an isomorphism of function spaces.

\begin{lemma}Fix $\widehat{A}=0$ and $s>\frac{n}{2}+1$. Let $B_\delta(0)$ be a ball of sufficiently small radius $\delta$ containing $(\k,u,\E,A)$ that satisfy the Einstein-Yang-Mills system. Let $Y^\parallel$ be a vector field satisfying \eqref{Y-constraint}. The sum $X+Y^\parallel$ will then has the norm bound
\begin{align}
\hnorm{X+Y^\parallel}{s+1}&\leq C\big\{\hnorm{\k}{s-1}^2+\hnorm{u}{s}^2+e^{-2T}(\hnorm{\E}{s-1}^2+\hnorm{A}{s}^2)+e^{-4T}(\hnorm{\E}{s-1}^2+\hnorm{A}{s}^2)^2\big\},
\end{align}

\noindent for a constant $C$ depending only on the background geometry.
\end{lemma}

\begin{proof}
Define the second-order elliptic operator $\mathcal{P}_{g,\gamma}:X^i\mapsto\Delta_\gamma^gX^i-g^{k\ell}R[\gamma]^i{}_{kj\ell}X^j$. Similarly, let $\mathcal{P}_{\gamma,\gamma}$ be given by
\begin{equation*}
\mathcal{P}_{\gamma,\gamma}:X^i\mapsto-\gamma^{k\ell}\nabla[\gamma]_k\nabla[\gamma]_\ell X^i-\gamma^{k\ell}R[\gamma]^i{}_{kj\ell}X^j.
\end{equation*}

\noindent As discussed in Section 2.2 and Lemma 6.1 of \cite{Andersson2009}, we have that $\mathcal{P}_{\gamma,\gamma}$ is an isomorphism of suitable function spaces, and when applied to $X+Y^\parallel$ gives
\begin{align*}
\mathcal{P}_{\gamma,\gamma}\big(X^i+{Y^i}^\parallel\big)&=-2\k^{ij}\nabla_jN+\bigg(1-\frac{2}{n}\bigg)\nabla^iN-2\tau_0^2e^{-2T}Ng^{ik}g^{\ell j}F^a{}_{bk\ell}\E^b{}_{aj}+2N\k^{jk}\big(\Gamma[g]^i{}_{jk}-\Gamma[\gamma]^i{}_{jk}\big) \\
&\quad+(g^{mn}-\gamma^{mn})\partial_T\Gamma[\gamma]^i{}_{mn}+(g^{mn}-\gamma^{mn})\nabla[\gamma]_m\nabla[\gamma]_nX^i+(g^{mn}-\gamma^{mn})R[\gamma]^i{}_{mjn}X^j.
\end{align*}

\noindent Using elliptic regularity and the product estimates for $s>\frac{n}{2}+1$, we have that $X+Y^\parallel$ is bounded by
\begin{align*}
\hnorm{X+Y^\parallel}{s+1}&\leq C\hnorm{\mathcal{P}_{\gamma,\gamma}(X+Y^\parallel)}{s-1} \\
&\leq C\big\{\hnorm{\omega}{s}(1+\hnorm{\k}{s-1})+C\hnorm{u}{s}\big\{\norm{\partial_T\gamma}+\hnorm{X}{s+1}+\hnorm{X}{s-1}\big\} \\
&\quad+(1+\hnorm{\omega}{s-1})(e^{-2T}\hnorm{F}{s-1}\hnorm{\E}{s-1}+\hnorm{\k}{s-1}\hnorm{u}{s})\big\} \\
&\leq C\big\{\hnorm{\omega}{s+1}(1+\hnorm{\k}{s-1})+C\hnorm{u}{s}(1+\hnorm{u}{s})\big\{\hnorm{\omega}{s+1}+\hnorm{X}{s+1})\big\} \\
&\quad+(1+\hnorm{\omega}{s+1})(e^{-2T}\hnorm{A}{s}\hnorm{\E}{s-1}(1+\hnorm{A}{s})+\hnorm{\k}{s-1}\hnorm{u}{s})\big\}
\end{align*}

\noindent Applying Lemmas \ref{omega-estimate} and \ref{X-estimate} then yields the desired estimate for $X+Y^\parallel$.
\end{proof}

\section{Energy Estimates}\label{energies}
With the elliptic estimates from the previous section, we may now present an energy estimate argument to show that the perturbed data will decay. However we first need to examine the operators $\mathcal{L}_{\gamma_0,\gamma_0}$ and $\mathfrak{L}_{\gamma_0,\gamma_0}$, as they will play a key role in our analysis. Recall that $\mathcal{L}_{\gamma_0,\gamma_0}$ is assumed to have non-negative spectrum and $\mathfrak{L}_{\gamma_0,\gamma_0}$ will always have non-negative spectra. Let $\lambda_\mathcal{L}$ and $\lambda_\mathfrak{L}$ denote the smallest non-zero eigenvalues of these two operators, respectively. We will then define
\begin{align}
\lambda_E&=\lambda_\mathcal{L}-\epsilon, \\
\lambda_Y&=\lambda_\mathfrak{L}-\epsilon,
\end{align}

\noindent for some $\epsilon>0$ such that $\lambda_E$ and $\lambda_Y$ are both strictly positive. As the spectra of $\mathcal{L}_{\gamma,\gamma}$ and $\mathfrak{L}_{\gamma,\gamma}$ will evolve as $g$ evolves, and hence the shadow metric $\gamma$, evolves and the smallest non-zero eigenvalue plays a critical role in obtaining the decay of our energies, we must choose $\lambda_E$ and $\lambda_Y$ to be smaller than $\lambda_\mathcal{L}$ and $\lambda_\mathfrak{L}$, respectively, such that the decay we obtain is uniform. Due to the small data scenario we are working in, if we choose a sufficiently large $\epsilon$ and sufficiently small neighborhood $B_\delta(0)$ to work in, then we can guarantee that no $\mathcal{L}_{g,\gamma}$ or $\mathfrak{L}_{g,\gamma}$ has eigenvalues below $\lambda_E$ and $\lambda_Y$, respectively.

We now may define the first-order energy $E^{(1)}=E^{(1)}_\mathrm{Ein}+E^{(1)}_\mathrm{YM}+c_E\Gamma^{(1)}_\mathrm{Ein}-c_Y\Gamma^{(1)}_\mathrm{YM}$, where
\begin{align}
E^{(1)}_\mathrm{Ein}&=\langle\k,\k\rangle+\frac{1}{4}\langle u,\mathcal{L}_{g,\gamma}u\rangle, \\
E_\mathrm{YM}^{(1)}&=\langle\E,\E\rangle+\langle A-\widehat{A},\mathfrak{L}_{g,\gamma}(A-\widehat{A})\rangle, \\
\Gamma^{(1)}_\mathrm{Ein}&=\frac{1}{n}\langle\k, u\rangle, \\
\Gamma^{(1)}_\mathrm{YM}&=\frac{1}{n}\langle\E, A-\hat{A}\rangle,
\end{align}

\noindent and $c_E,c_Y>0$ are constants given by
\begin{align}
c_E&=\min\bigg\{\frac{2n^2\lambda_E}{(n-1)},\;\frac{n-1}{2}\bigg\} \\
c_Y&=\min\bigg\{\frac{\sqrt{\lambda_Y}}{2},\;\frac{n-3}{2}\bigg\}.
\end{align}

\noindent We will see later that the correction terms $\Gamma^{(1)}_\mathrm{Ein}$ and $\Gamma^{(1)}_\mathrm{YM}$ will be vital in obtaining a decay of the energy. We will also see that these corrections do not impact the positive-definiteness of $E^{(1)}$ due to the choice of $c_E$ and $c_Y$, which will be bounded from below by the non-negative Sobolev norms of $(\k,u,\E,A)$. As such, the decay of the energy will imply, as desired, that the Sobolev norms of $(\k,u,\E,A)$ will vanish.

From the first-order energy, we may construct higher-order energies. To do so, we make repeated application of the operators $\mathcal{L}_{g,\gamma}^i$ and $\mathfrak{L}_{g,\gamma}^i$ to the small values $(\k,u,\E,A)$ in the first-order energy. That is, the higher-order energy terms will be given by
\begin{align*}
E^{(i)}&=E^{(i)}_\mathrm{Ein}+E^{(i)}_\mathrm{YM}+c_E\Gamma^{(i)}_\mathrm{Ein}-c_Y\Gamma^{(i)}_\mathrm{YM}, \\
E^{(i)}_\mathrm{Ein}&=\big\langle\k,\mathcal{L}_{g,\gamma}^{i-1}\k\big\rangle+\frac{1}{4}\big\langle u,\mathcal{L}_{g,\gamma}^iu\big\rangle, \\
E^{(i)}_\mathrm{YM}&=\big\langle\E,\mathfrak{L}_{g,\gamma}^{i-1}\E\big\rangle+\big\langle(A-\widehat{A}),\mathfrak{L}_{g,\gamma}^i(A-\widehat{A})\big\rangle, \\
\Gamma^{(i)}_\mathrm{Ein}&=\frac{1}{n}\big\langle\k,\mathcal{L}_{g,\gamma}^{i-1}u\big\rangle, \\
\Gamma^{(i)}_\mathrm{YM}&=\frac{1}{n}\big\langle\E,\mathfrak{L}_{g,\gamma}^{i-1}(A-\widehat{A})\big\rangle.
\end{align*}

\noindent With $s>\frac{n}{2}+1$, the total energy will then be given by
\begin{equation}\label{total-energy}
E_s=\sum_{i=1}^sE^{(i)}.
\end{equation}

\noindent However we will not need to study the total energy in its entirety. Rather, due to the embeddings $H^{s_1}\hookrightarrow H^{s_2}$ for $s_1>s_2$ and $H^s\hookrightarrow L^2$ for all $s$, we have that the $L^2$ norms of the derivatives $(\nabla[\gamma]^I\k,\nabla[\gamma]^Iu,\nabla[\gamma]^I\E,\nabla[\gamma]^IA)$ are controlled by the same Sobolev norms as $(\k,u,\E,A)$. Indeed, higher-order energies will only additional contribute terms at higher-order in the small data due to the need to commute derivatives and apply Proposition \ref{important-ineqs}.2. With this in mind, to study the leading order behavior of $E_s$, we see that it will suffice to estimate $E^{(1)}$ by the appropriate small data norms.

We will denote $E_\mathrm{Ein}=\sum_{i=1}^sE^{(i)}_\mathrm{Ein}$ and $E_\mathrm{YM}=\sum_{i=1}^sE^{(i)}_\mathrm{YM}$, and have, by the immediately preceding discussion, the equivalence
\begin{align*}
E_\mathrm{Ein}&\approx\sum_{i=1}^s\lnorm{\nabla[\gamma]^{i-1}\k}{2}^2+\lnorm{\nabla[\gamma]^iu}{2}^2\approx\hnorm{\k}{s-1}^2+\hnorm{u}{s}^2, \\
E_\mathrm{YM}&\approx\sum_{i=1}^s\lnorm{\nabla[\gamma]^{i-1}\E}{2}^2+\lnorm{\nabla[\gamma]^iA}{2}^2\approx\hnorm{\E}{s-1}^2+\hnorm{A}{s}^2.
\end{align*}

\noindent where $f\approx g$ if there is a constant $C$ such that $C^{-1}f\leq g\leq Cf$.

To begin our analysis, we will estimate the time evolution of the energy and show that, to leading order, the energy must decay. It will be convenient to study the gravitational and Yang-Mills separately, so we compute that the time derivative of $E^{(1)}_\mathrm{Ein}$ is given as
{\allowdisplaybreaks[1]
\begin{align*}
\partial_TE^{(1)}_\mathrm{Ein}&=\underbrace{2\int_\Sigma\gamma^{im}\gamma^{jn}\partial_T\k_{ij}\k_{mn}\:\mu_g+\frac{1}{2}\int_\Sigma\gamma^{im}\gamma^{jn}\partial_Tu_{ij}\mathcal{L}_{g,\gamma}u_{mn}\:\mu_g}_{\text{Type I}_\mathrm{Ein}} \\
&+\underbrace{\frac{1}{4}\int_\Sigma\gamma^{im}\gamma^{jn}u_{ij}[\partial_T,\mathcal{L}_{g,\gamma}]u_{mn}\:\mu_g}_{\text{Type II}_\mathrm{Ein}} \\
&+\underbrace{2\int_\Sigma\partial_T\gamma^{im}\gamma^{jn}\k_{ij}\k_{mn}\:\mu_g+\frac{1}{2}\int_\Sigma\partial_T\gamma^{im}\gamma^{jn}u_{ij}\mathcal{L}_{g,\gamma}u_{mn}\:\mu_g}_{\text{Type III}_\mathrm{Ein}} \\
&+\underbrace{\frac{1}{2}\int_\Sigma\gamma^{im}\gamma^{jn}\k_{ij}\k_{mn}\tr(\partial_Tg)\:\mu_g+\frac{1}{8}\int_\Sigma\gamma^{im}\gamma^{jn}u_{ij}\mathcal{L}_{g,\gamma}u_{mn}\tr(\partial_Tg)\:\mu_g}_{\text{Type IV}_\mathrm{Ein}},
\end{align*}
}

\noindent while the time derivative of $E^{(1)}_\mathrm{YM}$ is given by
\begin{align*}
\partial_TE^{(1)}_\mathrm{YM}&=\underbrace{2\int_\Sigma\gamma^{ij}\partial_T\E^a{}_{bi}\E^b{}_{aj}\:\mu_g+2\int_\Sigma\gamma^{ij}\partial_T(A-\widehat{A})^a{}_{bi}\mathfrak{L}_{g,\gamma}(A-\widehat{A})^b{}_{aj}\:\mu_g}_{\text{Type I}_\mathrm{YM}} \\
&+\underbrace{\int_\Sigma\gamma^{ij}(A-\widehat{A})^a{}_{bi}[\partial_T,\mathfrak{L}_{g,\gamma}](A-\widehat{A})^b{}_{aj}\:\mu_g}_{\text{Type II}_\mathrm{YM}} \\
&+\underbrace{\int_\Sigma\partial_T\gamma^{ij}\E^a{}_{bi}\E^b{}_{aj}\:\mu_g+\int_\Sigma\partial_T\gamma^{ij}(A-\widehat{A})^a{}_{bi}\mathfrak{L}_{g,\gamma}(A-\widehat{A})^b{}_{aj}\:\mu_g}_{\text{Type III}_\mathrm{YM}} \\
&+\underbrace{\frac{1}{2}\int_\Sigma\gamma^{ij}\E^a{}_{bi}\E^b{}_{aj}\tr(\partial_Tg)\:\mu_g+\frac{1}{2}\int_\Sigma\gamma^{ij}(A-\widehat{A})^a{}_{bi}\mathfrak{L}_{g,\gamma}(A-\widehat{A})^b{}_{aj}\tr(\partial_Tg)\:\mu_g}_{\text{Type IV}_\mathrm{YM}}.
\end{align*}

\noindent Meanwhile, the time evolution of the correction terms $\Gamma^{(1)}_\mathrm{Ein}$ and $\Gamma^{(1)}_\mathrm{YM}$ will be given by
\begin{align*}
\partial_T\Gamma^{(1)}_\mathrm{Ein}&=\frac{1}{n}\int_\Sigma\gamma^{im}\gamma^{jn}\partial_T\k_{ij}u_{mn}\:\mu_g+\frac{1}{n}\int_\Sigma\gamma^{im}\gamma^{jn}\k_{ij}\partial_Tu_{mn}\:\mu_g \\
&\quad+\frac{2}{n}\int_\Sigma\partial_T\gamma^{im}\gamma^{jn}\k_{ij}u_{mn}\:\mu_g+\frac{1}{2n}\int_\Sigma\gamma^{im}\gamma^{jn}\k_{ij}u_{mn}\tr(\partial_Tg)\:\mu_g \\
\partial_T\Gamma^{(1)}_\mathrm{YM}&=\frac{1}{n}\int_\Sigma\gamma^{ij}\partial_T\E^a{}_{bi}A^b{}_{aj}\:\mu_g+\frac{1}{n}\int_\Sigma\gamma^{ij}\E^a{}_{bi}\partial_TA^b{}_{aj}\:\mu_g \\
&\quad+\frac{1}{n}\int_\Sigma\partial_T\gamma^{ij}\E^a{}_{bi}A^b{}_{aj}\:\mu_g+\frac{1}{2}\int_\Sigma\gamma^{ij}\E^a{}_{bi}A^b{}_{aj}\tr(\partial_Tg)\:\mu_g.
\end{align*}

\noindent The time derivatives of the small data $(\k,u,\E,A)$ are given by the evolution equations for $(\k,g,\E,A)$. In particular, with the time derivative of the background metric $\gamma$---an element of the vector space $T_\gamma\mathcal{N}$---able to be decomposed as in \eqref{N-tan-space}, we find that the time evolution of $u$ is given by
\begin{align}
\partial_Tu_{ij}&=2(\omega+n)\k_{ij}-h^{\TT\parallel}-X^m\nabla[\gamma]_mu_{ij}+\mathcal{F}^{(1)}_{ij} \\
\mathcal{F}^{(1)}_{ij}&=\frac{2}{n}\omega u_{ij}-u_{im}\nabla[\gamma]_jX^m-u_{jm}\nabla[\gamma]_iX^m+\frac{2}{n}\omega\gamma_{ij}-L_{X+Y^\parallel}\gamma_{ij},
\end{align}

\noindent where the principle terms that may be problematic for regularity and the leading-order decay terms have been separated from the remaining terms $\mathcal{F}^{(1)}_{ij}$.

It will also be convenient to separate out such terms in the evolution equations for $\k$, $\E$, and $A$. Doing so, we write the evolution equation for $\k$ as
{\allowdisplaybreaks[1]
\begin{align}
\partial_T\k_{ij}&=-(n-1)\k_{ij}-\frac{1}{2}(\omega+n)\mathcal{L}_{g,\gamma}u_{ij}-X^m\nabla[\gamma]_m\k_{ij}+\mathcal{F}^{(2)}_{ij} \\
\mathcal{F}^{(2)}_{ij}&=-\frac{n-2}{n}\omega\k_{ij}-\k_{im}\nabla[\gamma]_jX^m-\k_{jm}\nabla[\gamma]_iX^m \notag \\
&\qquad+\nabla[\gamma]_i\nabla[\gamma]_j\omega-\frac{1}{2}g^{mn}\big(\nabla[\gamma]_iu_{jn}+\nabla[\gamma]_ju_{in}-\nabla[\gamma]_nu_{ij}\big)\nabla[\gamma]_m\omega \notag \\
&\qquad+\tau_0^2e^{-2T}(\omega+n)\bigg\{\E^a{}_{bi}\E^b{}_{aj}-F^a{}_{bik}F^b{}_{aj}{}^k+\bigg(\frac{1}{2(n-1)}|F|^2-\frac{1}{n-1}|\E|^2\bigg)(u_{ij}+\gamma_{ij})\bigg\} \notag \\
&\qquad-\frac{\omega}{n^2}(u_{ij}+\gamma_{ij})+(\omega+n)\big\{2\k_{ik}\k_j{}^k-J_{ij}\big\} \\
J_{ij}&=\frac{1}{2}(u_{i\ell}u^{mn}R[\gamma]^\ell{}_{mjn}+u_{j\ell}u^{mn}R[\gamma]^\ell{}_{min}) \notag \\
&\qquad+\frac{1}{2}g^{mn}\gamma^{k\ell}\bigg\{\nabla[\gamma]_ju_{nk}\nabla[\gamma]_\ell u_{im}+\nabla[\gamma]_iu_{\ell m}\nabla[\gamma]_ku_{jn}-\frac{1}{2}\nabla[\gamma]_ju_{nk}\nabla[\gamma]_iu_{\ell m} \notag \\
&\qquad+\nabla[\gamma]_mu_{i\ell}\nabla[\gamma]_nu_{jk}-\nabla[\gamma]_mu_{i\ell}\nabla[\gamma]_ku_{jn}\bigg\}.
\end{align}
}

\noindent For $A$, we have
\begin{align}
\partial_TA^a{}_{bi}&=-(\omega+n)\E^a{}_{bi}-X^j\nabla[\gamma]_jA^a{}_{bi}+\mathcal{F}^{(3)}{}^a{}_{bi} \\
\mathcal{F}^{(3)}{}^a{}_{bi}&=-\nabla[\gamma]_i\phi^a{}_b-[A_i,\phi]^a{}_b-A^a{}_{bj}\nabla[\gamma]_iX^j.
\end{align}

\noindent Finally, for $\E$ we write the evolution equation in the form
\begin{align}
\partial_T\E^a{}_{bi}&=-(n-3)\E^a{}_{bi}+(\omega+n)\mathfrak{L}_{g,\gamma}A^a{}_{bi}-X^k\nabla[\gamma]_k\E^a{}_{bi}+\mathcal{F}^{(4)}{}^a{}_{bi} \\
\mathcal{F}^{(4)}{}^a{}_{bi}&=-\frac{n-2}{n}\omega\E^a{}_{bi}+g^{jk}\nabla[\gamma]_j\omega\big(\nabla[\gamma]_iA^a{}_{bk}-\nabla[\gamma]_kA^a{}_{bi}+[A_i,A_k]^a{}_b\big) \notag \\
&\qquad+\E^a{}_{bj}\big(2(\omega+n)\k_i{}^j-\nabla[\gamma]_iX^j\big)-[\E_i,\phi]^a{}_b+(\omega+n)S^a{}_{bi} \\
S^a{}_{bi}&=g_{im}g^{j\ell}\nabla[\gamma]_jg^{mk}\big(\nabla[\gamma]_kA^a{}_{b\ell}-\nabla[\gamma]_\ell A^a{}_{bk}+[A_k,A_\ell]^a{}_b\big) \notag \\
&\qquad+\nabla[\gamma]_jg^{j\ell}\big(\nabla[\gamma]_iA^a{}_{b\ell}-\nabla[\gamma]_\ell A^a{}_{bi}+[A_i,A_\ell]^a{}_b\big) \notag \\
&\qquad+g^{j\ell}\big(2[\nabla[\gamma]_jA_i,A_\ell]^a{}_b+[A_i,\nabla[\gamma]_jA_\ell]^a{}_b-[\nabla[\gamma]_iA_j,A_\ell]^a{}_b\big) \notag \\
&\qquad-g^{j\ell}\big[[A_i,A_j],A_\ell\big]^a{}_b+\nabla[\gamma]_i\big(g^{j\ell}\nabla[\gamma]_j\widehat{A}^a{}_{b\ell}-g^{j\ell}[\widehat{A}_j,A_\ell-\widehat{A}_\ell]^a{}_b\big).
\end{align}

\subsection{Yang-Mills Sector}\label{YM-sector-estimates}
\subsubsection{Time Evolution of Energy}\label{YM-sector-time-evol}
Recall we fix $\widehat{A}\equiv0$. With this condition imposed, we begin by estimating the Type $\mathrm{I}_\mathrm{YM}$ terms. Explicitly, we find that these terms are given by
\begin{align*}
\mathrm{I}_\mathrm{YM}&=2\langle\partial_T\E,\E\rangle+2\langle\partial_TA,\mathfrak{L}_{g,\gamma}A\rangle \\
&=-2(n-3)\langle\E,\E\rangle+2\langle(\omega+n)\mathfrak{L}_{g,\gamma}A,\E\rangle-2\langle X^k\nabla[\gamma]_k\E,\E\rangle+2\langle\mathcal{F}^{(4)},\E\rangle \\
&\quad-2\langle(\omega+n)\E,\mathfrak{L}_{g,\gamma}A\rangle-2\langle X^k\nabla[\gamma]_kA,\mathfrak{L}_{g,\gamma}A\rangle+2\langle\mathcal{F}^{(3)},\mathfrak{L}_{g,\gamma}A\rangle \\
&=-2(n-3)\langle\E,\E\rangle-2\langle X^k\nabla[\gamma]_k\E,\E\rangle+2\langle\mathcal{F}^{(4)},\E\rangle \\
&\quad-2\langle X^k\nabla[\gamma]_kA,\mathfrak{L}_{g,\gamma}A\rangle+2\langle\mathcal{F}^{(3)},\mathfrak{L}_{g,\gamma}A\rangle.
\end{align*}

\noindent We see that from the choice of operator $\mathfrak{L}_{g,\gamma}$ to act on $A$, the principle terms cancel point-wise. We will now also show that only $\langle\E,\E\rangle$ gives a leading order decay.

First we must handle the terms containing derivatives of $\E$ and $A$ in the direction of the shift $X$. Because the total energy is constructed by repeatedly applying the second-order differential operator $\mathfrak{L}_{g,\gamma}$ to the small data, additional derivatives of $\E$ and $A$ will be problematic as they cannot be controlled by the appropriate Sobolev norm for the highest-order energies. As such, we must show that the terms containing additional derivatives of $\E$ and $A$ can indeed be controlled by appropriate norms, a fact which will hold true due to the higher regularity of $X$. In particular, using integration by parts, we may explicitly compute $\langle X^k\nabla[\gamma]_k\E,\E\rangle$ as
\begin{align*}
\langle X^k\nabla[\gamma]_k\E,\E\rangle&=\int_\Sigma\gamma^{ij}X^k\nabla[\gamma]_k\E^a{}_{bi}\E^b{}_{aj}\:\mu_g \\
\langle X^k\nabla[\gamma]_k\E,\E\rangle&=-\int_\Sigma\gamma^{ij}X^k\E^a{}_{bi}\nabla[\gamma]_k\E^b{}_{aj}\:\mu_g \\
&\quad-\int_\Sigma\gamma^{ij}\big(\nabla[\gamma]_kX^k+X^k(\Gamma[g]^\ell{}_{k\ell}-\Gamma[\gamma]^\ell{}_{k\ell})\big)\E^a{}_{bi}\E^b{}_{aj}\:\mu_g \\
2\langle X^k\nabla[\gamma]_k\E,\E\rangle&=-\int_\Sigma\gamma^{ij}\big(\nabla[\gamma]_kX^k+X^k(\Gamma[g]^\ell{}_{k\ell}-\Gamma[\gamma]^\ell{}_{k\ell})\big)\E^a{}_{bi}\E^b{}_{aj}\:\mu_g \\
\langle X^k\nabla[\gamma]_k\E,\E\rangle&=-\frac{1}{2}\int_\Sigma\gamma^{ij}\big(\nabla[\gamma]_kX^k+X^k(\Gamma[g]^\ell{}_{k\ell}-\Gamma[\gamma]^\ell{}_{k\ell})\big)\E^a{}_{bi}\E^b{}_{aj}\:\mu_g.
\end{align*}

\noindent Here we also use that the covariant derivative with respect to the background metric $\gamma$ of the measure $\mu_g$ is given by $(\Gamma[g]-\Gamma[\gamma])^\ell{}_{k\ell}\mu_g$. We may then estimate the inner product as
\begin{align}
\big|\langle X^k\nabla[\gamma]_k\E,\E\rangle\big|&\leq C\big(\lnorm{\nabla[\gamma]X}{\infty}+\lnorm{X}{\infty}\lnorm{\Gamma[g]-\Gamma[\gamma]}{\infty}\big)\lnorm{\E}{2}^2 \notag \\
&\leq C\big(\hnorm{\nabla[\gamma]X}{s-1}+\hnorm{X}{s-1}\hnorm{\Gamma[g]-\Gamma[\gamma]}{s-1}\big)\hnorm{\E}{s-1}^2, \notag \\
&\leq C\big(\hnorm{X}{s+1}+\hnorm{X}{s+1}\hnorm{u}{s}\big)\hnorm{\E}{s-1}^2
\end{align}

\noindent where we use the embeddings $H^s\hookrightarrow H^{s-1}\hookrightarrow L^\infty$ for $s-1>\frac{n}{2}$, and that to leading order $\hnorm{\Gamma[g]-\Gamma[\gamma]}{s-1}=\hnorm{u}{s}$. Due to the higher regularity of $X$, we are thus still able to control the inner product.

We may perform a similar computation for the term containing a derivative of $A$ in the direction of $X$, giving
\begin{align*}
\langle X^k\nabla&[\gamma]_kA,\mathfrak{L}_{g,\gamma}A\rangle \\
&=\int_\Sigma\gamma^{ij}X^k\nabla[\gamma]_kA^a{}_{bi}\mathfrak{L}_{g,\gamma}A^b{}_{aj}\:\mu_g \\
&=-\int_\Sigma\gamma^{ij}X^k\nabla[\gamma]_kA^a{}_{bi}\big(\nabla[\gamma]_m(g^{mn}\mu_g\nabla[\gamma]_nA^b{}_{aj})+g^{\ell m}R[\gamma]^n{}_{m\ell j}A^b{}_{an}\:\mu_g\big) \\
&=\int_\Sigma\gamma^{ij}(\nabla[\gamma]_mX^k\nabla[\gamma]_kA^a{}_{bi}+X^k\nabla[\gamma]_m\nabla[\gamma]_kA^a{}_{bi})g^{mn}\nabla[\gamma]_nA^b{}_{aj}\:\mu_g \\
&\quad-\int_\Sigma\gamma^{ij}X^k\nabla[\gamma]_kA^a{}_{bi}g^{\ell m}R[\gamma]^n{}_{m\ell j}A^b{}_{an}\:\mu_g \\
&=\int_\Sigma\gamma^{ij}g^{mn}X^k\nabla[\gamma]_k\nabla[\gamma]_mA^a{}_{bi}\nabla[\gamma]_nA^b{}_{aj}\:\mu_g+\int_\Sigma\gamma^{ij}g^{mn}X^kR[\gamma]^\ell{}_{imk}A^a{}_{b\ell}\nabla[\gamma]_nA^b{}_{aj}\:\mu_g \\
&\quad+\int_\Sigma\gamma^{ij}g^{mn}\nabla[\gamma]_mX^k\nabla[\gamma]_kA^a{}_{bi}\nabla[\gamma]_nA^b{}_{aj}\:\mu_g-\int_\Sigma\gamma^{ij}X^k\nabla[\gamma]_kA^a{}_{bi}g^{\ell m}R[\gamma]^n{}_{m\ell j}A^b{}_{an}\:\mu_g \\
&=\frac{1}{2}\int_\Sigma\gamma^{ij}g^{mn}X^k\nabla[\gamma]_k(\nabla[\gamma]_mA^a{}_{bi}\nabla[\gamma]_nA^b{}_{aj})\:\mu_g+\int_\Sigma\gamma^{ij}g^{mn}X^kR[\gamma]^\ell{}_{imk}A^a{}_{b\ell}\nabla[\gamma]_nA^b{}_{aj}\:\mu_g \\
&\quad+\int_\Sigma\gamma^{ij}g^{mn}\nabla[\gamma]_mX^k\nabla[\gamma]_kA^a{}_{bi}\nabla[\gamma]_nA^b{}_{aj}\:\mu_g-\int_\Sigma\gamma^{ij}X^k\nabla[\gamma]_kA^a{}_{bi}g^{\ell m}R[\gamma]^n{}_{m\ell j}A^b{}_{an}\:\mu_g \\
&=-\frac{1}{2}\int_\Sigma\gamma^{ij}\big(\nabla[\gamma]_kg^{mn}X^k+g^{mn}\nabla[\gamma]_kX^k+g^{mn}X^k(\Gamma[g]^\ell{}_{k\ell}-\Gamma[\gamma]^\ell{}_{k\ell})\big)\nabla[\gamma]_mA^a{}_{bi}\nabla[\gamma]_nA^b{}_{aj}\:\mu_g \\
&\quad+\int_\Sigma\gamma^{ij}g^{mn}X^kR[\gamma]^\ell{}_{imk}A^a{}_{b\ell}\nabla[\gamma]_nA^b{}_{aj}\:\mu_g \\
&\quad+\int_\Sigma\gamma^{ij}g^{mn}\nabla[\gamma]_mX^k\nabla[\gamma]_kA^a{}_{bi}\nabla[\gamma]_nA^b{}_{aj}\:\mu_g-\int_\Sigma\gamma^{ij}X^k\nabla[\gamma]_kA^a{}_{bi}g^{\ell m}R[\gamma]^n{}_{m\ell j}A^b{}_{an}\:\mu_g.
\end{align*}

\noindent Again, we use integration by parts and Stokes' theorem on a compact manifold to move the extraneous derivative on $A$ to one on $X$. Using the same embeddings as for $\E$, we may estimate $\langle X^k\nabla[\gamma]_kA,\mathfrak{L}_{g,\gamma}A\rangle$ in terms of the appropriate Sobolev norms. In particular, using the estimate for $X$ from Lemma \ref{X-estimate}, we get that the terms containing extra derivatives of $\E$ and $A$ may be controlled by
\begin{align}
\Big|\langle X^k\nabla[\gamma]_k\E,\E\rangle\Big|&\leq C\hnorm{\E}{s-1}^2(\hnorm{\k}{s-1}+\hnorm{u}{s}) \notag \\
&\quad+Ce^{-2T}\hnorm{\E}{s-1}^2(\hnorm{\E}{s-1}^2+\hnorm{A}{s}^2)+Ce^{-4T}\hnorm{\E}{s-1}^2(\hnorm{\E}{s-1}^2+\hnorm{A}{s})^2 \\
\Big|\langle X^k\nabla[\gamma]_kA,\mathfrak{L}_{g,\gamma}A\rangle\Big|&\leq C\hnorm{A}{s}^2(\hnorm{\k}{s-1}+\hnorm{u}{s}) \notag \\
&\quad+Ce^{-2T}\hnorm{A}{s}^2(\hnorm{\E}{s-1}^2+\hnorm{A}{s}^2)+Ce^{-4T}\hnorm{A}{s}^2(\hnorm{\E}{s-1}^2+\hnorm{A}{s}^2)^2
\end{align}

\noindent Meanwhile, using the Sobolev embeddings and the elliptic estimates for $\omega$, $X$, and $\phi$, we can also estimate the terms containing $\mathcal{F}^{(4)}$ and $\mathcal{F}^{(3)}$ by
\begin{align}
\big|\langle\mathcal{F}^{(4)},\E\rangle\big|&\leq C\hnorm{\E}{s-1}^2\big(\hnorm{\k}{s-1}+\hnorm{u}{s}\big)+C\hnorm{\E}{s-1}\hnorm{A}{s}\big(\hnorm{u}{s}+\hnorm{A}{s}\big) \notag \\
&+Ce^{-2T}(\hnorm{\E}{s-1}^2+\hnorm{A}{s}^2)^2+Ce^{-4T}\hnorm{\E}{s-1}^2(\hnorm{\E}{s-1}^2+\hnorm{A}{s}^2)^2 \\
\Big|\langle\mathcal{F}^{(3)},\mathfrak{L}_{g,\gamma}A\rangle\Big|&\leq C\hnorm{A}{s}(\hnorm{\k}{s-1}^2+\hnorm{u}{s}^2+\hnorm{\E}{s-1}^2+\hnorm{A}{s}^2) \notag \\
&+Ce^{-2}(\hnorm{\E}{s-1}^2+\hnorm{A}{s}^2)^2+Ce^{-4T}\hnorm{A}{s}^2(\hnorm{\E}{s-1}^2+\hnorm{A}{s}^2)^2
\end{align}

\noindent Note that each term in $\mathcal{F}^{(3)}$ has regularity at the level of $A$, and so we can integrate by parts to move derivatives from the second-order operator $\mathfrak{L}_{g,\gamma}A$ onto $\mathcal{F}^{(3)}$ and retain the appropriate regularity. Putting all the estimates together, we have that the Type $\mathrm{I}_\mathrm{YM}$ terms may be written as
\begin{equation}
\mathrm{I}_\mathrm{YM}=-2(n-3)\lnorm{\E}{2}^2+\mathcal{R}_\mathrm{YM},
\end{equation}

\noindent where $\mathcal{R}_\mathrm{YM}$ satisfies the estimate
\begin{align}\label{YM-remainder-estimate}
|\mathcal{R}_\mathrm{YM}|&\leq C(\hnorm{\E}{s-1}^2+\hnorm{A}{s}^2)(\hnorm{\k}{s-1}+\hnorm{u}{s}) \notag \\
&\quad+C\hnorm{A}{s}(\hnorm{\k}{s-1}^2+\hnorm{u}{s}^2+\hnorm{\E}{s-1}^2+\hnorm{A}{s}^2) \notag \\
&\quad+Ce^{-2T}(\hnorm{\E}{s-1}^2+\hnorm{A}{s}^2)^2 \notag \\
&\quad+Ce^{-4T}(\hnorm{\E}{s-1}^2+\hnorm{A}{s}^2)^3.
\end{align}

\noindent We will now see that, to leading order, the Type $\mathrm{I}_\mathrm{YM}$ terms will control the decay of $\partial_TE_\mathrm{YM}$. That is, the Type $\mathrm{II}_\mathrm{YM}$, $\mathrm{III}_\mathrm{YM}$, and $\mathrm{IV}_\mathrm{YM}$ terms will satisfy only third-order estimates. To begin, by differentiating under the integral sign and using that $\mathfrak{L}_{g,\gamma}$ is a self-adjoint operator for the given inner product, we may write the Type $\mathrm{II}_\mathrm{YM}$ terms as
\begin{align*}
\mathrm{II}_\mathrm{YM}&=\int_\Sigma\gamma^{ij}A^a{}_{bi}\partial_T\mathfrak{L}_{g,\gamma}A^b{}_{aj}\:\mu_g-\int_\Sigma\gamma^{ij}A^a{}_{bi}\mathfrak{L}_{g,\gamma}\partial_TA^b{}_{aj}\:\mu_g \\
&=\der{}{T}\int_\Sigma\gamma^{ij}A^a{}_{bi}\mathfrak{L}_{g,\gamma}A^b{}_{aj}\:\mu_g-\int_\Sigma\gamma^{ij}\partial_TA^a{}_{bi}\mathfrak{L}_{g,\gamma}A^b{}_{aj}\:\mu_g-\int_\Sigma\gamma^{ij}A^a{}_{bi}\mathfrak{L}_{g,\gamma}\partial_TA^b{}_{aj}\:\mu_g \\
&\quad-\int_\Sigma\partial_T\gamma^{ij}A^a{}_{bi}\mathfrak{L}_{g,\gamma}A^b{}_{aj}\:\mu_g-\frac{1}{2}\int_\Sigma\gamma^{ij}A^a{}_{bi}\mathfrak{L}_{g,\gamma}A^b{}_{aj}\tr(\partial_Tg)\:\mu_g \\
&=\der{}{T}\int_\Sigma\gamma^{ij}A^a{}_{bi}\mathfrak{L}_{g,\gamma}A^b{}_{aj}\:\mu_g-2\int_\Sigma\gamma^{ij}A^a{}_{bi}\mathfrak{L}_{g,\gamma}\partial_TA^b{}_{aj}\:\mu_g \\
&\quad-\int_\Sigma\partial_T\gamma^{ij}A^a{}_{bi}\mathfrak{L}_{g,\gamma}A^b{}_{aj}\:\mu_g-\frac{1}{2}\int_\Sigma\gamma^{ij}A^a{}_{bi}\mathfrak{L}_{g,\gamma}A^b{}_{aj}\tr(\partial_Tg)\:\mu_g.
\end{align*}

\noindent Observe that the final two terms of $\mathrm{II}_\mathrm{YM}$ will cancel point-wise with terms in $\mathrm{III}_\mathrm{YM}$ and $\mathrm{IV}_\mathrm{YM}$. We will thus only look at how to control the first two terms in the above expression. In particular, writing out the definition of $\mathfrak{L}_{g,\gamma}$ and differentiating under the integral sign, we have
\begin{align*}
\der{}{T}\int_\Sigma\gamma^{ij}A^a{}_{bi}\mathfrak{L}_{g,\gamma}A^b{}_{aj}\:\mu_g&=\der{}{T}\int_\Sigma\gamma^{ij}A^a{}_{bi}\big(\Delta^\gamma_gA^b{}_{aj}-g^{k\ell}R[\gamma]^m{}_{\ell jk}A^b{}_{am}\big)\:\mu_g \\
&=\der{}{T}\int_\Sigma\gamma^{ij}g^{k\ell}\big(\nabla[\gamma]_kA^a{}_{bi}\nabla[\gamma]_\ell A^b{}_{aj}-A^a{}_{bi}R[\gamma]^m{}_{\ell jk}A^b{}_{am}\big)\:\mu_g \\
&=\int_\Sigma\partial_T\gamma^{ij}g^{k\ell}\big(\nabla[\gamma]_kA^a{}_{bi}\nabla[\gamma]_\ell A^b{}_{aj}-A^a{}_{bi}R[\gamma]^m{}_{\ell jk}A^b{}_{am}\big)\:\mu_g \\
&\quad+\int_\Sigma\gamma^{ij}\partial_Tg^{k\ell}\big(\nabla[\gamma]_kA^a{}_{bi}\nabla[\gamma]_\ell A^b{}_{aj}-A^a{}_{bi}R[\gamma]^m{}_{\ell jk}A^b{}_{am}\big)\:\mu_g \\
&\quad+\frac{1}{2}\int_\Sigma\gamma^{ij}g^{k\ell}\big(\nabla[\gamma]_kA^a{}_{bi}\nabla[\gamma]_\ell A^b{}_{aj}-A^a{}_{bi}R[\gamma]^m{}_{\ell jk}A^b{}_{am}\big)\tr(\partial_Tg)\:\mu_g \\
&\quad+2\int_\Sigma\gamma^{ij}g^{k\ell}\big[\partial_T,\nabla[\gamma]_k\big]A^a{}_{bi}\nabla[\gamma]_\ell A^b{}_{aj}\:\mu_g+2\int_\Sigma\gamma^{ij}g^{k\ell}\nabla[\gamma]_kA^a{}_{bi}\nabla[\gamma]_\ell\partial_TA^b{}_{aj}\:\mu_g \\
&\quad-2\int_\Sigma\gamma^{ij}g^{k\ell}A^a{}_{bi}R[\gamma]^m{}_{\ell jk}\partial_TA^b{}_{am}\:\mu_g-\int_\Sigma\gamma^{ij}g^{k\ell}A^a{}_{bi}\partial_TR[\gamma]^m{}_{\ell jk}A^b{}_{am}\:\mu_g \\
&=\int_\Sigma\partial_T\gamma^{ij}A^a{}_{bi}\mathfrak{L}_{g,\gamma}A^b{}_{aj}\:\mu_g+\frac{1}{2}\int_\Sigma\gamma^{ij}A^a{}_{bi}\mathfrak{L}_{g,\gamma}A^b{}_{aj}\tr(\partial_Tg)\:\mu_g \\
&\quad+\int_\Sigma\gamma^{ij}\partial_Tg^{k\ell}\big(\nabla[\gamma]_kA^a{}_{bi}\nabla[\gamma]_\ell A^b{}_{aj}-A^a{}_{bi}R[\gamma]^m{}_{\ell jk}A^b{}_{am}\big)\:\mu_g \\
&\quad+2\int_\Sigma\gamma^{ij}g^{k\ell}\big[\partial_T,\nabla[\gamma]_k\big]A^a{}_{bi}\nabla[\gamma]_\ell A^b{}_{aj}\:\mu_g \\
&\quad+2\int_\Sigma\gamma^{ij}A^a{}_{bi}\mathfrak{L}_{g,\gamma}\partial_TA^b{}_{aj}\:\mu_g-\int_\Sigma\gamma^{ij}g^{k\ell}A^a{}_{bi}\partial_TR[\gamma]^m{}_{\ell jk}A^b{}_{am}\:\mu_g.
\end{align*}

\noindent Collecting all of the Type $\mathrm{II}_\mathrm{YM}$, $\mathrm{III}_\mathrm{YM}$, and $\mathrm{IV}_\mathrm{YM}$ terms together, we find that
\begin{align*}
\mathrm{II}_\mathrm{YM}+\mathrm{III}_\mathrm{YM}+\mathrm{IV}_\mathrm{YM}&=\frac{1}{2}\int_\Sigma\big(2\partial_T\gamma^{ij}+\gamma^{ij}\tr(\partial_Tg)\big)\big(\E^a{}_{bi}\E^b{}_{aj}+A^a{}_{bi}\mathfrak{L}_{g,\gamma}A^b{}_{aj}\big)\:\mu_g \\
&\quad+\int_\Sigma\gamma^{ij}\partial_Tg^{k\ell}\nabla[\gamma]_kA^a{}_{bi}\nabla[\gamma]_\ell A^b{}_{aj}\:\mu_g+2\int_\Sigma\gamma^{ij}g^{k\ell}\big[\partial_T,\nabla[\gamma]_k\big]A^a{}_{bi}\nabla[\gamma]_\ell A^b{}_{aj}\:\mu_g \\
&\quad-\int_\Sigma\gamma^{ij}g^{k\ell}A^a{}_{bi}\partial_TR[\gamma]^m{}_{\ell jk}A^b{}_{am}\:\mu_g.
\end{align*}

\noindent Using the embedding $H^{s-1}\hookrightarrow L^\infty$ for $s-1>\frac{n}{2}$, we then have the estimate
\begin{align*}
\big|\mathrm{II}_\mathrm{YM}+\mathrm{III}_\mathrm{YM}+\mathrm{IV}_\mathrm{YM}\big|&\leq C\big\{\hnorm{\E}{s-1}^2+\hnorm{A}{s}^2\big\}\big\{\norm{\partial_T\gamma}+\hnorm{\partial_Tg}{s-1}\big\} \\
&\quad+C\hnorm{A}{s}\hnorm{[\partial_T,\nabla[\gamma]]A}{s-1}+C\hnorm{A}{s}^2\norm{\partial_T\mathrm{Riem}[\gamma]}.
\end{align*}

\noindent Observe that $[\partial_T,\nabla[\gamma]]A=(\partial_T\Gamma[\gamma])A$ because partial derivatives commute, and so we may write
\begin{equation*}
\hnorm{[\partial_T,\nabla[\gamma]]A}{s-1}=\hnorm{(D\Gamma[\gamma].\partial_T\gamma)\cdot A}{s-1}\leq C\norm{\partial_T\gamma}\hnorm{A}{s-1},
\end{equation*}

\noindent where $(D\Gamma[\gamma].\partial_T\gamma)\cdot A$ is the Fr\'echet derivative of $\Gamma[\gamma]$ in the direction of $\partial_T\gamma$, contracted with $A$. We may similarly write $\partial_T\mathrm{Riem}[\gamma]=D\mathrm{Riem}[\gamma].\partial_T\gamma$ to see that $\norm{\partial_T\mathrm{Riem[\gamma]}}\leq C\norm{\partial_T\gamma}$. Using also the evolution equations for $g$ and Lemma \ref{gamma-evol}, we finally obtain the estimate
\begin{align}
\big|\mathrm{II}_\mathrm{YM}+\mathrm{III}_\mathrm{YM}+\mathrm{IV}_\mathrm{YM}\big|&\leq C\big\{\hnorm{\E}{s-1}^2+\hnorm{A}{s}^2\big\} \notag \\
&\quad\times\big\{\hnorm{\k}{s-1}+\hnorm{u}{s}+e^{-2T}(\hnorm{\E}{s-1}^2+\hnorm{A}{s}^2)+e^{-4T}(\hnorm{\E}{s-1}^2+\hnorm{A}{s}^2)^2\big\}
\end{align}

\noindent We see then that the Type $\mathrm{II}_\mathrm{YM}$, $\mathrm{III}_\mathrm{YM}$, and $\mathrm{IV}_\mathrm{YM}$ terms in $\partial_TE_\mathrm{YM}$ will satisfy a third-order estimate. In fact, if we increase the constant in \eqref{YM-remainder-estimate}, we will have that the non-Type $\mathrm{I}_\mathrm{YM}$ terms satisfy the same estimate as the Type $\mathrm{I}_\mathrm{YM}$ remainder $\mathcal{R}_\mathrm{YM}$. We may thus write $\partial_TE_\mathrm{YM}$ as
\begin{equation}\label{EYM-first-order}
\partial_TE^{(1)}_\mathrm{YM}=-2(n-3)\lnorm{\E}{2}^2+\mathcal{R}_\mathrm{YM},
\end{equation}

\noindent with $\mathcal{R}_\mathrm{YM}$ satisfying bound given in \eqref{YM-remainder-estimate} for a re-defined constant $C$. Note that we may write the estimate on $\mathcal{R}_\mathrm{YM}$ in terms of $E_\mathrm{YM}\approx\hnorm{\E}{s-1}^2+\hnorm{A}{s}^2$ and $E_\mathrm{Ein}\approx\hnorm{\k}{s-1}^2+\hnorm{u}{s}^2$, giving us
\begin{equation}\label{YM-R-estimate}
|\mathcal{R}_\mathrm{YM}|\leq C\big\{(E_\mathrm{YM})^\frac{3}{2}+(E_\mathrm{Ein})^\frac{1}{2}E_\mathrm{YM}+E_\mathrm{Ein}(E_\mathrm{YM})^\frac{1}{2}+e^{-2T}(E_\mathrm{YM})^2+e^{-4T}(E_\mathrm{YM})^3\big\}.
\end{equation}

\begin{remark}
The failure of the energy estimate argument for $n=3$ spatial dimensions manifests in \eqref{EYM-first-order}. In particular we see that the leading-order decay term will vanish, which will prevent us from obtaining a uniformly decaying bound on the total energy \eqref{total-energy} and hence on the Sobolev norms of the small data. As mentioned earlier, this is a consequence of the conformal invariance of Yang-Mills equations in $3+1$ dimensions.
\end{remark}

Finally, we must find an appropriate estimate for the time derivative of the correction term $\Gamma^{(1)}_\mathrm{YM}$. Using the embedding of $H^{s-1}$ into $L^\infty$ for $s>\frac{n}{2}+1$, we write $\partial_T\Gamma^{(1)}_\mathrm{YM}$ as
\begin{equation}
\partial_T\Gamma^{(1)}_\mathrm{YM}=-\frac{(n-3)}{n}\langle\E,A\rangle+\langle\mathfrak{L}_{g,\gamma}A,A\rangle-\langle\E,\E\rangle+\mathcal{S}_\mathrm{YM},
\end{equation}

\noindent where $\mathcal{S}_\mathrm{YM}$ satisfies the third-order estimate
\begin{align*}
|\mathcal{S}_\mathrm{YM}|&\leq C\big\{\hnorm{\omega}{s+1}(\hnorm{\E}{s-1}^2+\hnorm{A}{s}^2)+\hnorm{\E}{s-1}\hnorm{A}{s}(\hnorm{X}{s+1}+\hnorm{\partial_Tg}{s-1}) \\
&\qquad\qquad\qquad\qquad\qquad\qquad\qquad\quad+\hnorm{\E}{s-1}\hnorm{\mathcal{F}^{(3)}}{s-1}+\hnorm{A}{s}\hnorm{\mathcal{F}^{(4)}}{s-1}\big\} \\
&\leq C\hnorm{\E}{s-1}\big\{\hnorm{\k}{s-1}^2+\hnorm{u}{s}^2+\hnorm{\E}{s-1}^2+\hnorm{A}{s}^2\big\} \\
&\quad+C\hnorm{\E}{s-1}\hnorm{A}{s}\big\{\hnorm{\k}{s-1}+\hnorm{u}{s}\big\}+C\hnorm{A}{s}^2(\hnorm{u}{s}+\hnorm{A}{s}) \\
&\quad+Ce^{-2T}\big\{\hnorm{\E}{s-1}^2+\hnorm{A}{s}^2\big\}^2+Ce^{-4T}\big\{\hnorm{\E}{s-1}^2+\hnorm{A}{s}^2\big\}^3 \\
&\leq C\big\{(E_\mathrm{YM})^\frac{3}{2}+(E_\mathrm{Ein})^\frac{1}{2}E_\mathrm{YM}+E_\mathrm{Ein}(E_\mathrm{YM})^\frac{1}{2}+e^{-2T}(E_\mathrm{YM})^2+e^{-4T}(E_\mathrm{YM})^3\big\}.
\end{align*}

\noindent Note that $\mathcal{S}_\mathrm{YM}$ satisfies the same estimate as $\mathcal{R}_\mathrm{YM}$, albeit with a different constant $C$. Putting together the two expressions for $\partial_TE^{(1)}_\mathrm{YM}$ and $\partial_T\Gamma^{(1)}_\mathrm{YM}$, we arrive at the total contribution from the Yang-Mills sector to $\partial_TE^{(1)}$,
\begin{equation*}
\partial_T\big(E_\mathrm{YM}^{(1)}-c_Y\Gamma_\mathrm{YM}^{(1)}\big)=-\big(2(n-3)-c_Y\big)\langle\E,\E\rangle-c_Y\langle\mathfrak{L}_{g,\gamma}A,A\rangle+\frac{(n-3)}{n}c_Y\langle\E,A\rangle+\mathcal{R}_\mathrm{YM},
\end{equation*}

\noindent for a redefined $\mathcal{R}_\mathrm{YM}$ that still satisfies the estimate \eqref{YM-R-estimate}. To show that the total energy decays, we will want to show that, to leading order, the time evolution is uniformly bounded as $\partial_TE_s\leq-\alpha E_s$ for some positive constant $\alpha$. We will now see that the Yang-Mills time evolution independently satisfies a bound of this form, using a similar analysis to that in Lemma 6.4 of \cite{Andersson2009}.

To begin, let $Y=\frac{c_Y}{n-3}$, such that $0<Y<1$, and define $\alpha_Y=(n-3)\big[1-\sqrt{1-Y}\big]$. We may then write
\begin{align*}
\partial_T\big(E^{(1)}_\mathrm{YM}-c_Y\Gamma^{(1)}_\mathrm{YM}\big)&=-\alpha_Y\big(E^{(1)}_\mathrm{YM}-c_Y\Gamma^{(1)}_\mathrm{YM}\big)-\big(2(n-3)-c_Y-\alpha_Y\big)\langle\E,\E\rangle \\
&\qquad-(c_Y-\alpha_Y)\langle\mathfrak{L}_{g,\gamma}A,A\rangle+\frac{c_Y}{n}\big((n-3)-\alpha_Y\big)\langle\E,A\rangle+\mathcal{R}_\mathrm{YM},
\end{align*}

\noindent For the time evolution of $E^{(1)}_\mathrm{YM}-c_Y\Gamma^{(1)}_\mathrm{YM}$ to satisfy the desired inequality to leading order, we must have that
\begin{equation}\label{need-neg-YM-energy}
-\big(2(n-3)-c_Y-\alpha_Y\big)\langle\E,\E\rangle-(c_Y-\alpha_Y)\langle\mathfrak{L}_{g,\gamma}A,A\rangle+\frac{c_Y}{n}\big((n-3)-\alpha_Y\big)\langle\E,A\rangle<0.
\end{equation}

\noindent Since $\Sigma$ is compact, $\mathfrak{L}_{g,\gamma}$ will have discrete spectrum. If we perform a spectral decomposition of $\mathfrak{L}_{g,\gamma}$, we then have the requirement that
\begin{align*}
\sum_{k=1}^\infty-\big(2(n-3)-c_Y-\alpha_Y\big)\big\langle\langle\E,e_k\rangle,\langle\E,e_k\rangle\big\rangle&-(c_Y-\alpha_Y)\lambda_k\big\langle\langle A,e_k\rangle,\langle A,e_k\rangle\big\rangle \\
&+\frac{c_Y}{n}\big((n-3)-\alpha_Y\big)\big\langle\langle\E,e_k\rangle e_k,\langle A,e_k\rangle e_k\rangle<0,
\end{align*}

\noindent where $\lambda_k>0$ and $e_k$ are the $k^\text{th}$ eigenvalue and eigentensor of $\mathfrak{L}_{g,\gamma}$. However, we see that each term in the summation is equal to the quadratic form
\begin{equation*}
M_k=\begin{bmatrix}c_Y+\alpha_Y-2(n-3) & \frac{c_Y}{2n}\big((n-3)-\alpha_Y\big) \\ \frac{c_Y}{2n}\big((n-3)-\alpha_Y\big) & -\lambda_k(c_Y-\alpha_Y)\end{bmatrix}=(n-3)\begin{bmatrix}Y-1-\sqrt{1-Y} & -\frac{Y}{2n}\sqrt{1-Y} \\ -\frac{Y}{2n}\sqrt{1-Y} & -\lambda_k(Y-1+\sqrt{1-Y})\end{bmatrix}
\end{equation*}

\noindent acting on $\big((\langle\E,e_k\rangle e_k,\langle A,e_k\rangle e_k),(\langle\E,e_k\rangle e_k,\langle A,e_k\rangle e_k)\big)$. Hence we will have \eqref{need-neg-YM-energy} hold if and only if $M_k$ is negative definite for each $k$. We may compute the trace and determinant of $M_k$ as
\begin{align*}
\operatorname{Tr}M_k&=(Y-1-\sqrt{1-Y})-\lambda_k\sqrt{1-Y}(1-\sqrt{1-Y}), \\
\det M_k&=(1-Y)Y\bigg[\lambda_k-\frac{c_Y}{4n^2(n-3)}\bigg].
\end{align*}

\noindent Since $0<Y<1$, we have that $\operatorname{Tr}M_k<0$. Since $n>3$ and $c_Y<\sqrt{\lambda_Y}\leq\lambda_k$, we have $\det M_k>0$. Hence each $M_k$ is negative definite, and so we get the following lemma.

\begin{lemma}\label{YM-time-evol}
Fix $\widehat{A}=0$ and $s>\frac{n}{2}+1$. Let $B_\delta(0)$ be a ball of sufficiently small radius $\delta$ containing $(\k,u,\E,A)$ that satisfy the Einstein-Yang-Mills system. Then the time evolution of $E^{(1)}_\mathrm{YM}-c_Y\Gamma^{(1)}_\mathrm{YM}$ is given by
\begin{equation}
\partial_T\big(E^{(1)}_\mathrm{YM}-c_Y\Gamma^{(1)}_\mathrm{YM}\big)=-\big(2(n-3)-c_Y\big)\langle\E,\E\rangle-c_Y\langle\mathfrak{L}_{g,\gamma}A,A\rangle+\frac{(n-3)}{n}c_Y\langle\E,A\rangle+\mathcal{R}_\mathrm{YM},
\end{equation}

\noindent where $\mathcal{R}_\mathrm{YM}$ satisfies the third-order estimate
\begin{equation}
|\mathcal{R}_\mathrm{YM}|\leq C\big\{(E_\mathrm{YM})^\frac{3}{2}+(E_\mathrm{Ein})^\frac{1}{2}E_\mathrm{YM}+E_\mathrm{Ein}(E_\mathrm{YM})^\frac{1}{2}+e^{-2T}(E_\mathrm{YM})^2+e^{-4T}(E_\mathrm{YM})^3\big\}
\end{equation}

\noindent for a constant $C$ depending only on the background geometry and $\tau_0$. In particular, the time evolution of the Yang-Mills energies will satisfy the bound
\begin{equation}\label{YM-time-evol-bound}
\partial_T\big(E^{(1)}_\mathrm{YM}-c_Y\Gamma^{(1)}_\mathrm{YM}\big)\leq-\alpha_Y\big(E^{(1)}_\mathrm{YM}-c_Y\Gamma^{(1)}_\mathrm{YM}\big)+\mathcal{R}_\mathrm{YM},
\end{equation}

\noindent where
\begin{equation}\label{YM-alpha-def}
\alpha_Y=(n-3)\bigg(1-\sqrt{1-\frac{c_Y}{n-3}}\bigg).
\end{equation}
\end{lemma}

\subsubsection{Positive-Definiteness of Energy}
Having the time evolution of the Yang-Mills sector, we also wish to show that $E_\mathrm{YM}-c_Y\Gamma_\mathrm{YM}$ is positive definite, and in particular is bounded from below by the Sobolev norms of the small data $\E$ and $A$. Doing so will allow us to later bound these norms by the total energy, which we will show will decay and thus force $\hnorm{\E}{s-1}$ and $\hnorm{A}{s}$ to zero. We give the result of this section in the following lemma.

\begin{lemma}\label{YM-pos-def}
Fix $\widehat{A}=0$ and $s>\frac{n}{2}+1$. Then there is a $\delta>0$ and a constant $C>0$, such that $(\k,u,\E,A)\in B_\delta(0)$ satisfy the inequality
\begin{equation}
\hnorm{\E}{s-1}^2+\hnorm{A}{s}^2\leq C\sum_{i=1}^s(E^{(i)}_\mathrm{YM}-c_Y\Gamma^{(i)}_\mathrm{YM}).
\end{equation}
\end{lemma}

\begin{proof}
To begin, note that $(\k,u,\E,A)=(0,0,0,0)$ is a fixed point of the total energy $E_s$ and hence the Yang-Mills energy $\sum_{i=1}^sE^{(i)}_\mathrm{YM}-c_Y\Gamma^{(i)}_\mathrm{YM}$. Since the Yang-Mills energy is a smooth functional of $\E$ and $A$, to determine if it is positive-definite we only need to see if the Hessian at $(\k,g,\E,A)=(0,\gamma,0,0)$ is positive-definite. The Hessian at this fixed point may be written in the form
\begin{equation*}
D^2(E^{(i)}_\mathrm{YM}-c_Y\Gamma^{(i)}_\mathrm{YM})\big((\E,A),(\E,A)\big)=2\langle\E,\mathfrak{L}_{\gamma,\gamma}^{i-1}\E\rangle+2\langle A,\mathfrak{L}_{\gamma,\gamma}^iA\rangle-\frac{4c_Y}{n}\langle\E,\mathfrak{L}_{\gamma,\gamma}^{i-1}A\rangle
\end{equation*}

\noindent Performing a spectral decomposition of $\mathfrak{L}_{\gamma,\gamma}$, we then see that the Hessian will be positive definite if and only if the quadratic form
\begin{equation*}
M_k=\begin{bmatrix}1 & -c_Y/n \\ -c_Y/n & \lambda_k\end{bmatrix}
\end{equation*}

\noindent is. However, this follows directly from the construction of $c_Y$ and the fact that $\mathfrak{L}_{\gamma,\gamma}$ has non-negative spectrum. We see that $\operatorname{Tr}M=1+\lambda_k>0$ and $\det M=\lambda_k-\frac{c_Y^2}{n^2}>0$, as $c_Y$ has been chosen such that $c_Y<\sqrt{\lambda_Y}<\sqrt{\lambda_k}$. Thus $M$ is positive definite and so, it follows, is the Hessian of the $i^\text{th}$ Yang-Mills energy. Since this is true for all $1\leq i\leq s$, we then have that $\sum_{i=1}^sD^2(E^{(i)}_\mathrm{YM}-c_Y\Gamma^{(i)}_\mathrm{YM})\geq0$. However, because each of the quadratic forms $M_k$ are positive definite and thus invertible, we have that $\sum_{i=1}^sD^2(E^{(i)}_\mathrm{YM}-c_Y\Gamma^{(i)}_\mathrm{YM})$ will give an isomorphism of subsets of function spaces. From this isomorphism, we may conclude that for some constant $C$,
\begin{equation}\label{YM-sub-bound}
\hnorm{\E}{s-1}^2+\hnorm{A}{s}^2\leq C\sum_{i=1}^sD^2(E^{(i)}_\mathrm{YM}-c_Y\Gamma^{(i)}_\mathrm{YM}).
\end{equation}

\noindent Finally, we may apply the Morse-Palais lemma, giving there exists a neighborhood of the origin $B_\delta(0)$ such that, up to a diffeomorphism $B_\delta(0)\rightarrow B_\delta(0)$, we have
\begin{equation}\label{YM-Morse}
\sum_{i=1}^s(E^{(i)}_\mathrm{YM}-c_Y\Gamma^{(i)}_\mathrm{YM})=\sum_{i=1}^sD^2(E^{(i)}_\mathrm{YM}-c_Y\Gamma^{(i)}_\mathrm{YM}).
\end{equation}

\noindent Combining \eqref{YM-sub-bound} and \eqref{YM-Morse}, we obtain the desired estimate.
\end{proof}

Having obtained an estimate on the time evolution of the Yang-Mills energy, and seeing that the Yang-Mills energy bounds the Sobolev norms of $\E$ and $A$, we are now ready to study the gravitational sector and prove similar results for $\k$ and $u$.

\subsection{Gravitational Sector}\label{Ein-sector-estimates}
\subsubsection{Time Evolution of Energy}
We will now study the time evolution of $\partial_TE^{(1)}_\mathrm{Ein}$. In particular, we will see that Type $\mathrm{I}_\mathrm{Ein}$ terms control, to leading order, the decay of $\partial_TE^{(1)}_\mathrm{Ein}$, with the other type terms contributing only to higher order. Many of the calculations are done explicitly in \cite{Mondal2022}, with only the estimates changing here due to the differing evolution equations and elliptic estimates. As such, we will present only the necessary equations and estimates here, and refer readers to the relevant sections of \cite{Mondal2022} for additional calculations when appropriate.

To begin, we may write the Type $\mathrm{I}_\mathrm{Ein}$ terms as
\begin{align*}
\mathrm{I}_\mathrm{Ein}&=2\langle\partial_T\k,\k\rangle+\frac{1}{2}\langle\partial_tu,\mathcal{L}_{g,\gamma}u\rangle \\
&=-2(n-1)\langle\k,\k\rangle-\langle(\omega+n)\mathcal{L}_{g,\gamma}u,\k\rangle-2\langle X^m\nabla[\gamma]_m\k,\k\rangle+2\langle\mathcal{F}^{(2)},\k\rangle \\
&\quad+\langle(\omega+n)\k,\mathcal{L}_{g,\gamma}u\rangle-\frac{1}{2}\langle h^{\TT\parallel},\mathcal{L}_{g,\gamma}u\rangle-\frac{1}{2}\langle X^m\nabla[\gamma]_mu,\mathcal{L}_{g,\gamma}u\rangle+\langle\mathcal{F}^{(1)},\mathcal{L}_{g,\gamma}u\rangle \\
&=-2(n-1)\langle\k,\k\rangle-2\langle X^m\nabla[\gamma]_m\k,\k\rangle+2\langle\mathcal{F}^{(2)},\k\rangle \\
&\quad-\frac{1}{2}\langle h^{\TT\parallel},\mathcal{L}_{g,\gamma}u\rangle-\frac{1}{2}\langle X^m\nabla[\gamma]_mu,\mathcal{L}_{g,\gamma}u\rangle+\frac{1}{2}\langle\mathcal{F}^{(1)},\mathcal{L}_{g,\gamma}u\rangle,
\end{align*}

\noindent where we see the principle terms cancel point-wise. As in the Yang-Mills sector, we have potentially problematic terms containing additional derivatives of $\k$ and $u$. However, using integration by parts, we can move the extra derivative shift field $X\in H^{s+1}$. Doing so, we obtain
\begin{equation*}
\langle X^k\nabla[\gamma]_k\k,\k\rangle=-\frac{1}{2}\int_\Sigma\gamma^{ij}\gamma^{mn}\big(\nabla[\gamma]_kX^k+X^k(\Gamma[g]^\ell{}_{k\ell}-\Gamma[\gamma]^\ell{}_{k\ell})\big)\k_{in}\k_{jm}\:\mu_g
\end{equation*}

\noindent and
\begin{align*}
\langle X^k\nabla&[\gamma]_ku,\mathcal{L}_{g,\gamma}u\rangle \\
&=-\frac{1}{2}\int_\Sigma\gamma^{ij}\gamma^{mn}\big(\nabla[\gamma]_kg^{\ell r}X^k+g^{\ell r}\nabla[\gamma]_kX^k+g^{\ell r}X^k(\Gamma[g]^s{}_{ks}-\Gamma[\gamma]^s{}_{ks})\big)\nabla[\gamma]_\ell u_{in}\nabla[\gamma]_ru_{jm}\:\mu_g \\
&\quad+\int_\Sigma\gamma^{ij}\gamma^{mn}g^{\ell r}X^k\big(R[\gamma]^s{}_{i\ell k}u_{sn}+R[\gamma]^s{}_{n\ell r}u_{is}\big)\nabla[\gamma]_ru_{jm}\:\mu_g \\
&\quad+\int_\Sigma\gamma^{ij}\gamma^{mn}g^{\ell r}\nabla[\gamma]_\ell X^k\nabla[\gamma]_ku_{in}\nabla[\gamma]_ru_{jm}\:\mu_g-2\int_\Sigma\gamma^{ij}\gamma^{mn}X^k\nabla[\gamma]_ku_{in}g^{ab}g^{cd}\gamma_{js}R[\gamma]^s{}_{amc}u_{bd}\:\mu_g.
\end{align*}

\noindent Using the Sobolev embeddings $H^{s-1}\hookrightarrow L^\infty$ for $s-1>\frac{n}{2}$, we then have the estimates
\begin{align}
\Big|\langle X^k\nabla[\gamma]_k\k,\k\rangle\Big|&\leq C\big\{\hnorm{\k}{s-1}^2(\hnorm{\k}{s-1}+\hnorm{u}{s})\big\} \notag \\
&\quad+C\big\{e^{-2T}\hnorm{\k}{s-1}^2(\hnorm{\E}{s-1}^2+\hnorm{A}{s}^2)+e^{-4T}\hnorm{\k}{s-1}^2(\hnorm{\E}{s-1}^2+\hnorm{A}{s}^2)^2\big\} \\
\Big|\langle X^k\nabla[\gamma]_ku,\mathcal{L}_{g,\gamma}u\rangle\Big|&\leq C\big\{\hnorm{u}{s}^2(\hnorm{\k}{s-1}+\hnorm{u}{s})\big\} \notag \\
&\quad+C\big\{e^{-2T}\hnorm{u}{s}^2(\hnorm{\E}{s-1}^2+\hnorm{A}{s}^2)+e^{-4T}\hnorm{u}{s}^2(\hnorm{\E}{s-1}^2+\hnorm{A}{s}^2)^2\big\}
\end{align}

\noindent We must also estimate the term containing the transverse-traceless $h^{\TT\parallel}$, which is potentially problematic as we have no a priori bound for $h^{\TT\parallel}$. However by an explicit calculation, as done in equations (165)---(172) of \cite{Mondal2022}, we have that $h^{\TT\parallel}$ will be bounded by
\begin{equation}
\norm{h^{\TT\parallel}}\leq C\big\{\hnorm{\k}{s-1}+\hnorm{\omega}{s+1}\big\},
\end{equation}

\noindent where because $h^{\TT\parallel}$ is an element of the finite-dimensional deformation space, all its norms are equivalent. This also means that we need not worry about the extra derivative acting on $u$, as we may freely move it to $h^{\TT\parallel}$ via integration by parts to achieve the necessary regularity of $u$. Using the bound on $h^{\TT\parallel}$ and the Sobolev embeddings, we have the estimate 
\begin{equation}
\Big|\langle h^{\TT\parallel},\mathcal{L}_{g,\gamma}u\rangle\Big|\leq C\hnorm{u}{s}^2\big\{\hnorm{\k}{s-1}+e^{-2T}(\hnorm{\E}{s-1}^2+\hnorm{A}{s}^2)\big\}.
\end{equation}

\noindent To finish the estimates of the Type $\mathrm{I}_\mathrm{Ein}$ terms, we may estimate those containing $\mathcal{F}^{(2)}$ and $\mathcal{F}^{(1)}$ as
\begin{align}
\Big|\langle\mathcal{F}^{(2)},\k\rangle\Big|&\leq C\hnorm{\k}{s-1}\big\{\hnorm{\k}{s-1}^2+\hnorm{u}{s}^2\big\} \notag \\
&\quad+C\big\{e^{-2T}\hnorm{\k}{s-1}(\hnorm{\E}{s-1}^2+\hnorm{A}{s}^2)+e^{-4T}\hnorm{\k}{s-1}(\hnorm{\E}{s-1}^2+\hnorm{A}{s}^2)^2\big\}, \\
\Big|\langle\mathcal{F}^{(1)},\mathcal{L}_{g,\gamma}u\rangle\Big|&\leq C\hnorm{u}{s}\big\{\hnorm{\k}{s-1}^2+\hnorm{u}{s}^2\big\} \notag \\
&\quad+C\big\{e^{-2T}\hnorm{u}{s}(\hnorm{\E}{s-1}^2+\hnorm{A}{s}^2)+e^{-4T}\hnorm{u}{s}(\hnorm{\E}{s-1}^2+\hnorm{A}{s}^2)^2\big\}.
\end{align}

\noindent We note that each term in $\mathcal{F}^{(1)}$ has regularity at least at the level of $u$, and so we may integrate by parts to move excess derivatives from $\mathcal{L}_{g,\gamma}u$ onto $\mathcal{F}^{(1)}$ without any issue.

It now remains to estimate the Type $\mathrm{II}_\mathrm{Ein}$, $\mathrm{III}_\mathrm{Ein}$, and $\mathrm{IV}_\mathrm{Ein}$ terms. We may do so via a similar calculation to that done Section \ref{YM-sector-estimates}, as well as in \cite{Mondal2022}. We thus omit this computation, giving the result as
\begin{align*}
\mathrm{II}_\mathrm{Ein}+\mathrm{III}_\mathrm{Ein}+\mathrm{IV}_\mathrm{Ein}&=\frac{1}{8}\int_\Sigma\big(4\partial_T\gamma^{im}\gamma^{jn}+\gamma^{im}\gamma^{jn}\tr(\partial_Tg)\big)\big(4\k_{ij}\k_{mn}+u_{ij}\mathcal{L}_{g,\gamma}u_{mn}\big)\:\mu_g \\
&\quad+\frac{1}{4}\int_\Sigma\gamma^{im}\gamma^{jn}\partial_Tg^{k\ell}\nabla[\gamma]_ku_{ij}\nabla[\gamma]_\ell u_{mn}\:\mu_g \\
&\quad+\frac{1}{2}\int_\Sigma\gamma^{im}\gamma^{jn}g^{k\ell}\big[\partial_T,\nabla[\gamma]_k\big]u_{ij}\nabla[\gamma]_\ell u_{mn}\:\mu_g \\
&\quad-\frac{1}{2}\int_\Sigma\gamma^{im}\gamma^{jn}g^{kp}g^{\ell q}u_{ij}R[\gamma]_{mkn\ell}u_{pq}\:\mu_g.
\end{align*}

\noindent We may then obtain an estimate for the remaining terms of
\begin{align}
\big|\mathrm{II}_\mathrm{Ein}+\mathrm{III}_\mathrm{Ein}+\mathrm{IV}_\mathrm{Ein}\big|&\leq C\big\{\hnorm{\k}{s-1}^2+\hnorm{u}{s}^2\big\}\hnorm{\partial_Tg}{s-1} \notag \\
&\leq C\big\{\hnorm{\k}{s-1}^2+\hnorm{u}{s}^2\big\} \notag \\
&\quad\times\big\{\hnorm{\k}{s-1}+\hnorm{u}{s}+e^{-2T}(\hnorm{\E}{s-1}^2+\hnorm{A}{s}^2) \notag \\
&\qquad\qquad\qquad\qquad\qquad\;\,+e^{-4T}(\hnorm{\E}{s-1}^2+\hnorm{A}{s}^2)^2\big\}.
\end{align}

\noindent Combining all the estimates and keeping only the leading order terms as we are working in the small data regime, we find the time evolution of $E^{(1)}_\mathrm{Ein}$ is given by
\begin{equation}\label{E-Ein-eqn}
\partial_TE^{(1)}_\mathrm{Ein}=-2(n-1)\langle\k,\k\rangle+\mathcal{R}_\mathrm{Ein},
\end{equation}

\noindent where $\mathcal{R}_\mathrm{Ein}$ satisfies the bound
\begin{align}
|\mathcal{R}_\mathrm{Ein}|&\leq C\big\{\hnorm{\k}{s-1}+\hnorm{u}{s}\big\}\big\{\hnorm{\k}{s-1}^2+\hnorm{u}{s}^2\big\} \notag \\
&\qquad+C(\hnorm{\k}{s-1}+\hnorm{u}{s})\big\{e^{-2T}(\hnorm{\E}{s-1}^2+\hnorm{A}{s}^2)+e^{-4T}(\hnorm{\E}{s-1}^2+\hnorm{A}{s})^2\big\} \notag \\
&\leq C\big\{(E_\mathrm{Ein})^\frac{3}{2}+e^{-2T}(E_\mathrm{Ein})^\frac{1}{2}E_\mathrm{YM}+e^{-4T}(E_\mathrm{Ein})^\frac{1}{2}E_\mathrm{YM}^2\big\}.
\end{align}

\noindent We must finally estimate the time derivative of the correction to the gravitational energy, $\partial_T\Gamma^{(1)}_\mathrm{Ein}$. By an explicit calculation using the embeddings $H^s\hookrightarrow L^\infty$ for $s>\frac{n}{2}+1$ and $H^s\hookrightarrow L^2$ for any $s$, as well as standard inequalities, we find $\partial_T\Gamma^{(1)}_\mathrm{Ein}$ may be written in the form
\begin{equation}\label{Gamma-Ein-eqn}
\partial_T\Gamma^{(1)}_\mathrm{Ein}=-\frac{(n-1)}{n}\langle\k,u\rangle-\frac{1}{2}\langle\mathcal{L}_{g,\gamma}u,u\rangle+2\langle\k,\k\rangle+\mathcal{S}_\mathrm{Ein},
\end{equation}

\noindent where $\mathcal{S}_\mathrm{Ein}$ satisfies the estimate
\begin{align*}
|\mathcal{S}_\mathrm{Ein}|&\leq C\big\{\hnorm{\k}{s-1}+\hnorm{u}{s}\big\}\big\{\hnorm{\k}{s-1}^2+\hnorm{u}{s}^2+e^{-2T}(\hnorm{\E}{s-1}^2+\hnorm{A}{s}^2)+e^{-4T}(\hnorm{\E}{s-1}^2+\hnorm{A}{s}^2)^2\big\} \\
&\leq C\big\{(E_\mathrm{Ein})^\frac{3}{2}+e^{-2T}(E_\mathrm{Ein})^\frac{1}{2}E_\mathrm{YM}+e^{-4T}(E_\mathrm{Ein})^\frac{1}{2}E_\mathrm{YM}^2\big\}.
\end{align*}

\noindent Combining \eqref{E-Ein-eqn} and \eqref{Gamma-Ein-eqn}, we then have that the contribution of the gravitational sector to the time evolution of the first-order energy will be given as
\begin{equation*}
\partial_T(E^{(1)}_\mathrm{Ein}+c_E\Gamma^{(1)}_\mathrm{Ein})=-2\big((n-1)-c_E\big)\langle\k,\k\rangle-\frac{c_E}{2}\langle\mathcal{L}_{g,\gamma}u,u\rangle-\frac{(n-1)}{n}c_E\langle\k,u\rangle+\mathcal{R}_\mathrm{Ein}.
\end{equation*}

\noindent However, as with the Yang-Mills energy, we want to bound, to leading order, the time evolution of the gravitational sector energy by a term of the form $-\alpha_E(E^{(1)}_\mathrm{Ein}+c_E\Gamma^{(1)}_\mathrm{Ein})$. We note that \cite{Andersson2009} shows such a result, though the energies we have defined differ by a constant multiple factor of $\frac{1}{2n^2}$. A calculation analogous to that in Section \ref{YM-sector-time-evol} shows this does not change the result of Andersson and Moncrief, and so we obtain the following lemma.

\begin{lemma}\label{Ein-time-evol}
Let $s>\frac{n}{2}+1$ and take $B_\delta(0)$ to be a ball of sufficiently small radius $\delta$ containing $(\k,u,\E,A)$ that satisfy the Einstein-Yang-Mills system. Then the time evolution of $E^{(1)}_\mathrm{Ein}+c_E\Gamma^{(1)}_\mathrm{Ein}$ is given by
\begin{equation}
\partial_T(E^{(1)}_\mathrm{Ein}+c_E\Gamma^{(1)}_\mathrm{Ein})=-2\big((n-1)-c_E\big)\langle\k,\k\rangle-\frac{c_E}{2}\langle\mathcal{L}_{g,\gamma}u,u\rangle-\frac{(n-1)}{n}c_E\langle\k,u\rangle+\mathcal{R}_\mathrm{Ein},
\end{equation}

\noindent where $\mathcal{R}_\mathrm{Ein}$ satisfies the third-order estimate
\begin{equation}
|\mathcal{R}_\mathrm{Ein}|\leq C\big\{(E_\mathrm{Ein})^\frac{3}{2}+e^{-2T}(E_\mathrm{Ein})^\frac{1}{2}E_\mathrm{YM}+e^{-4T}(E_\mathrm{Ein})^\frac{1}{2}(E_\mathrm{YM})^2\big\}.
\end{equation}

\noindent In particular, the gravitational energy will be bounded by
\begin{equation}\label{Ein-time-evol-bound}
\partial_T(E^{(1)}_\mathrm{Ein}+c_E\Gamma^{(1)}_\mathrm{Ein})\leq-\alpha_E(E^{(1)}_\mathrm{Ein}+c_E\Gamma^{(1)}_\mathrm{Ein})+\mathcal{R}_\mathrm{Ein},
\end{equation}

\noindent for $\alpha_E$ defined as
\begin{equation}\label{Ein-alpha-def}
\alpha_E=(n-1)\bigg(1-\sqrt{1-\frac{2c_E}{n-1}}\bigg).
\end{equation}
\end{lemma}

\subsubsection{Positive-Definiteness of Energy}
As done for the Yang-Mills sector, we will need to show we have a bound of the form $\hnorm{\k}{s-1}^2+\hnorm{u}{s}^2\leq C\sum_{i=s}^s(E^{(i)}_\mathrm{Ein}+c_E\Gamma^{(i)}_\mathrm{Ein})$ so that a decay in the total energy will imply that the perturbations vanish. An estimate of this form, however, has been shown in Section 7.1 of \cite{Andersson2009}. While the energies we have defined differ by a positive constant factor of $\frac{1}{2n^2}$, we note that this will not affect the proof presented by Andersson and Moncrief. We thus restate here the result of Theorem 7.4 from their work, which will still hold, in the following form.

\begin{lemma}\label{Ein-pos-def}\cite{Andersson2009} Suppose that $\gamma$ satisfies the shadow metric condition for $g$, and take $s>\frac{n}{2}+1$. Then there is a $\delta>0$ and a constant $C>0$, such that for $(\k,u,\E,A)\in B_\delta(0)$, the inequality
\begin{equation*}
\hnorm{\k}{s-1}^2+\hnorm{u}{s}^2\leq C\sum_{i=s}^s(E^{(i)}_\mathrm{Ein}+c_E\Gamma^{(i)}_\mathrm{Ein}).
\end{equation*}

\noindent is satisfied.
\end{lemma}

\subsection{Total Energy}
Recall in \eqref{total-energy} that we have defined the total energy as the quantity
\begin{equation*}
E_s=\sum_{i=1}^sE^{(i)}=\sum_{i=s}^sE^{(i)}_\mathrm{Ein}+c_E\Gamma^{(i)}_\mathrm{Ein}+E^{(i)}_\mathrm{YM}-c_Y\Gamma^{(i)}_\mathrm{YM}.
\end{equation*}

\noindent Now, as a direct result of Lemmas \ref{YM-pos-def} and \ref{Ein-pos-def}, we obtain the following theorem related to the positive-definiteness of the total energy.
\begin{theorem}\label{total-E-pos-def}
Fix $s>\frac{n}{2}+1$. Suppose that $\gamma$ satisfies the shadow metric condition for $g$, and take $\widehat{A}=0$. Then there exists a $\delta>0$ and a constant $C>0$ such that, for $(\k,u,\E,A)\in B_\delta(0)$, the inequality
\begin{equation*}
\lnorm{\k}{s-1}^2+\hnorm{u}{s}^2+\hnorm{\E}{s-1}^2+\hnorm{A}{s}^2\leq CE_s
\end{equation*}

\noindent holds.
\end{theorem}

In light of Theorem \ref{total-E-pos-def}, if we can show, with sufficiently small initial data at $T_0$, that $E_s$ decays to zero as $T\rightarrow\infty$, then we will have shown that the perturbations also decay and hence the Einstein-Yang-Mills system converges to a background solution of the form \eqref{background-sol}. Furthermore, we will have that the maximal existence interval of the Cauchy problem for the Einstein-Yang-Mills system in the chosen gauges is then $[T_0,\infty)$, i.e., solutions of the small data Cauchy problem for the Einstein-Yang-Mills system will exist globally.

To show the decay of $E_s$, we first see that, by combining the results of Lemmas \ref{YM-time-evol} and \ref{Ein-time-evol}, we have that $\partial_TE^{(1)}$ is bounded from above by
\begin{align}
\partial_TE^{(1)}&\leq-\alpha_E\big(E^{(1)}_\mathrm{Ein}+c_E\Gamma^{(1)}_\mathrm{Ein}\big)-\alpha_Y\big(E^{(1)}_\mathrm{YM}-c_Y\Gamma^{(1)}_\mathrm{YM}\big)+\mathcal{R},
\end{align}

\noindent where $\mathcal{R}$ satisfies the leading-order estimate
\begin{align*}
|\mathcal{R}|&\leq|\mathcal{R}_\mathrm{YM}|+|\mathcal{R}_\mathrm{Ein}| \\
&\leq C\big\{\big((E_\mathrm{Ein})^\frac{1}{2}+(E_\mathrm{YM})^\frac{1}{2}\big)(E_\mathrm{Ein}+E_\mathrm{YM})+(E_\mathrm{Ein})^\frac{1}{2}\big(e^{-2T}E_\mathrm{YM}+e^{-4T}(E_\mathrm{YM})^2\big)\big\} \\
&\leq C\big\{(E_\mathrm{Ein}+E_\mathrm{YM})^\frac{3}{2}+(E_\mathrm{Ein}+E_\mathrm{YM})^\frac{1}{2}\big(e^{-2T}(E_\mathrm{Ein}+E_\mathrm{YM})+e^{-4T}(E_\mathrm{Ein}+E_\mathrm{YM})^2\big)\big\}.
\end{align*}

\noindent From Theorem \ref{total-E-pos-def} and the equivalence of $E_\mathrm{Ein}$ and $E_\mathrm{YM}$ with the Sobolev norms of the small data we have that
\begin{equation*}
E_\mathrm{Ein}+E_\mathrm{YM}\leq C\big\{\hnorm{\k}{s-1}^2+\hnorm{u}{s}^2+\hnorm{\E}{s-1}^2+\hnorm{A}{s}^2\big\}\leq CE_s.
\end{equation*}

\noindent This means that $\mathcal{R}$ will satisfy the bound
\begin{equation*}
|\mathcal{R}|\leq C\big\{(E_s)^\frac{3}{2}+e^{-2T}(E_s)^\frac{3}{2}+e^{-4T}(E_s)^\frac{5}{2}\big\}.
\end{equation*}

\noindent As we have finally obtained a bound for the higher-order terms on our total energy, we may make the estimate $e^{-T}\leq e^{-T_0}$ for $T>T_0$ without worry that dangerous terms of the form $e^T$ will appear. Using this, we then find the leading-order estimate of $\mathcal{R}$ to simply be
\begin{equation}\label{total-R-est}
|\mathcal{R}|\leq C(E_s)^\frac{3}{2}.
\end{equation}

\noindent Now, define $\Lambda=\frac{1}{2}\min\{\alpha_E,\alpha_Y\big\}$. Because both $E^{(1)}_\mathrm{Ein}+c_E\Gamma^{(1)}_\mathrm{Ein}$ and $E^{(1)}_\mathrm{YM}-c_Y\Gamma^{(1)}_\mathrm{YM}$ are positive-definite, we then have that the time evolution of the first-order total energy is bounded by
\begin{align}
\partial_TE^{(1)}&\leq-2\Lambda\big(E^{(1)}_\mathrm{Ein}+E^{(1)}_\mathrm{YM}+c_E\Gamma^{(1)}_\mathrm{Ein}-\Gamma^{(1)}_\mathrm{YM}\big)+\mathcal{R}=-\Lambda E^{(1)}+\mathcal{R}.
\end{align}

\noindent Furthermore, because $E^{(i)}(\k,u,\E,A)\approx E^{(1)}\big(\nabla[\gamma]^i\k,\nabla[\gamma]^iu,\nabla[\gamma]^i\E,\nabla[\gamma]^iA\big)$, we obtain an analogous inequality for each $E^{(i)}$. We remark that in obtaining these higher-order inequalities, we will need to commute differential operators in order to apply the same arguments as for the first-order energies. However by applying Proposition \ref{important-ineqs}.2, we will be able to bound these commutator terms by the appropriate Sobolev norms and avoid any issues with regularity; c.f. Section 2.1 of \cite{Andersson2003} or Section 5 of \cite{Mondal2021} for similar, detailed calculations.

These higher-order energy estimates, along with the estimate of $\mathcal{R}$ in \eqref{total-R-est}, leads us to the follow theorem.
\begin{theorem}
Fix $\widehat{A}=0$ and $s>\frac{n}{2}+1$. Let $B_\delta(0)$ be a ball of sufficiently small radius $\delta$ containing $(\k,u,\E,A)$ that satisfy the Einstein-Yang-Mills system. Fix $T_0>-\infty$. Then there exists a constant $C>0$ such that for $T\geq T_0$, the time evolution of $E_s$ satisfies the differential inequality
\begin{equation}
\partial_TE_s\leq-2\Lambda E_s+2C(E_s)^\frac{3}{2}.
\end{equation}
\end{theorem}

\noindent We may now integrate up this inequality. In particular, let $Y=(E_s)^\frac{1}{2}$. We then have the differential inequality
\begin{equation*}
\partial_TY\leq-\Lambda Y+CY^2.
\end{equation*}

\noindent Integrating up and substituting back $E_s$, we find that
\begin{equation}
E_s(T)\leq\Lambda^2\Bigg[C+\bigg(\frac{\Lambda}{\sqrt{E_s(T_0)}}-C\bigg)e^{\Lambda(T-T_0)}\Bigg]^{-2}.
\end{equation}

\noindent When $T>T_0$ and $E_s(T_0)<\frac{\Lambda^2}{C^2}$, we will have that $E_s(T)$ is bounded from above by an exponentially decaying function and from below by $0$. Hence $E_s$ must tend to zero as $T$ tends to infinity. We must simply require that the initial energy $E_s(T_0)$, determined by the initial values of $(\k,u,\E,A)$, is sufficiently small. However, with sufficiently small initial data, we have that the $E_s$ will be bounded for all $T\in[T_0,\infty)$ and thus, in light of Theorem \ref{total-E-pos-def}, the data $(\k,u,\E,A)$ will decay to a stable solution of the vacuum Einstein's equations. This leads us to the main theorem for this paper.
\begin{theorem}\label{main-theorem}
Suppose that $\gamma_0$ has a smooth deformation space $\mathcal{N}$ and that the operator $\mathcal{L}_{\gamma_0,\gamma_0}$ has non-negative spectrum. Fix $s>\frac{n}{2}+1$ for $n>3$. There is then a $\delta>0$ such that if $(\k_0,g_0-\gamma_0,\E_0,A_0)\in B_\delta(0)\subset H^{s-1}\times H^s\times H^{s-1}\times H^s$ with $\gamma_0$ a shadow metric of $g_0$, then the Cauchy problem for the Einstein-Yang-Mills system with initial data $(\k_0,g_0,\gamma_0,\E_0,A_0)$ is globally well-posed to the future.

In addition, if $T\mapsto(\k,g,\gamma,\E,A)$ is the maximal development to the Cauchy problem for the Einstein-Yang-Mills system, then there is a $\gamma_*\in\mathcal{N}$ such that $(\k,g,\gamma,\E,A)\rightarrow(0,\gamma_*,\gamma_*,0,0)$ as $T\rightarrow\infty$.
\end{theorem}

The main theorem also yields the following corollary, with an argument following that in \cite{Andersson2004}.
\begin{corollary}
Let $(\k_0,g_0,\gamma_0,\E_0,A_0)$ be as in Theorem \ref{main-theorem}, and let $(\tilde{\k}_0,\tilde{g}_0,\tilde{\gamma}_0,\tilde{\E}_0,\tilde{A}_0)$ be the corresponding unscaled Cauchy data. Let $(\tilde{\k},\tilde{g},\tilde{\gamma},\tilde{\E},\tilde{A})$ be the maximal development to the Cauchy problem for the unscaled Einstein-Yang-Mills system. Then the spacetime corresponding to $(\tilde{\k},\tilde{g},\tilde{\gamma},\tilde{\E},\tilde{A})$ is geodesically complete to the future.
\end{corollary}

\section{Gauge Covariant $3+1$ Yang-Mills Equations}\label{gauge-covariant}
As mentioned earlier, the energy estimates found in this paper will fail for $n=3$ spatial dimensions. This appeared to be a result of the choice of gauge made. However, we will presernt a gauge-covariant argument as to why such energy estimates will fail to yield the desired existence theorem.

First, let us define the gauge-covariant derivatives as follows 
\begin{align*}
\widehat{\mathcal{L}}_{\partial_{t}}&:=\partial_{t}+[A_{0},\cdot], \\
\widehat{\mathcal{L}}_X&:=\mathcal{L}_{X}+[X\cdot A,\cdot], \\ \widehat{\nabla}_i&:=\nabla_{i}+[A_{i},\cdot].
\end{align*}

\noindent Now, we define the chromo-electric and chromo-magnetic fields associated to the Yang-Mills field strength $F^a{}_{b\mu\nu}\mathrm{d}x^{\mu}\wedge\mathrm{d}x^{\nu}$ as
\begin{align*}
\E_i&:=F(\mathbf{n},\partial_{i}), \\
\mathcal{H}_i&:={}^{*}F(\mathbf{n},\partial_{i})=\frac{1}{2}\epsilon_{\mathbf{n}i}{}^{kl}F_{kl}=\frac{1}{2}\epsilon_i{}^{kl}F_{kl},
\end{align*}

\noindent where $\epsilon_{\mu\nu\alpha\beta}$ is the $4-$dimensional volume form. In terms of the chromo-electric and chromo-magnetic fields $\E$ and $\mathcal{H}$, the Yang-Mills equations take the following form in local ADM coordinates
\begin{align*}
\widehat{\mathcal{L}}_{\partial_t}\E_i&=\widehat{\mathcal{L}}_X\E_i-N\epsilon_i{}^{jk}\widehat{\nabla}_j\mathcal{H}_k-2Nk_i{}^j\E_j+N\tr(k)\E_i-\epsilon_i{}^{jk}\nabla_jN\mathcal{H}_k,\\
\widehat{\mathcal{L}}_{\partial_t}\mathcal{H}_i&=\widehat{\mathcal{L}}_X\mathcal{H}_i+N\epsilon_i{}^{jk}\widehat{\nabla}_j\E_k-2Nk_i{}^j\mathcal{H}_j+N\tr(k)\mathcal{H}_i+\epsilon_i{}^{jk}\nabla_jN\E_k
\end{align*}

\noindent which after re-scaling as in section \ref{rescale-eqns-section} reads in the dimensionless time $T$,
\begin{align*}
\widehat{\mathcal{L}}_{\partial_T}\E_i&=-\widehat{\mathcal{L}}_X\E_i+N\epsilon_i{}^{jk}\widehat{\nabla}[\gamma]_j\mathcal{H}_k+2N\k_i{}^{j}\E_j{}-{}\underbrace{\bigg(\frac{N}{3}-1\bigg)\E_i}_{\mathrm{I}}{}+{}\epsilon_i{}^{jk}\nabla_jN\mathcal{H}_k \\
&\qquad-N\epsilon_i{}^{jk}\big(\Gamma[g]^\ell{}_{jk}-\Gamma[\gamma]^\ell{}_{jk}\big)H_\ell, \\
\widehat{\mathcal{L}}_{\partial_T}\mathcal{H}_i&=-\widehat{\mathcal{L}}_X\mathcal{H}_i-N\epsilon_i{}^{jk}\widehat{\nabla}[\gamma]_j\E_k+2N\k_i{}^{j}\mathcal{H}_j{}-{}\underbrace{\bigg(\frac{N}{3}-1\bigg)\mathcal{H}_i}_{\mathrm{II}}{}-{}\epsilon_i{}^{jk}\nabla_jN\E_k \\
&\qquad+N\epsilon_i{}^{jk}\big(\Gamma[g]\ell{}_{jk}-\Gamma[\gamma]\ell{}_{jk}\big)\E_\ell.
\end{align*}

\noindent Notice that, in $3+1$ dimensions, we have $N=3+(\text{second order terms})$. Consequently, one loses the decay that may potentially arise from the terms $\mathrm{I}$ and $\mathrm{II}$. We define a gauge-invariant energy for the Yang-Mills fields
\begin{equation*}
E_\mathrm{YM}:=\frac{1}{2}\sum_{|I|\leq s-1}\langle\widehat{\nabla}^I\mathcal{E},\widehat{\nabla}^I\mathcal{E}\rangle+\langle \widehat{\nabla}^I\mathcal{H},\widehat{\nabla}^I\mathcal{H}\rangle.
\end{equation*}

\noindent Note that this is indeed gauge-invariant and for $s=1$, reduces to the standard Yang-Mills energy $\int_M\mathfrak{T}(\mathbf{n},\mathbf{n})$. We need compute up to $s=3$. This requires that we obtain several commutation relations, as terms of the form $[\widehat{\mathcal{L}}_{\partial_T},\widehat{\nabla}[\gamma]_i]$ will be present in the higher-order energies. Explicitly, we compute
\begin{align*}
[\widehat{\mathcal{L}}_{\partial_T},\widehat{\nabla}[\gamma]_i]\mathcal{R}_j&=[\widehat{\mathcal{L}}_{\partial_T},\widehat{\mathcal{L}}_{\partial_i}]\mathcal{R}_j+\partial_T\Gamma[\gamma]^\ell{}_{ij}\mathcal{R}_\ell \\
&=F(\partial_T,\partial_i)\mathcal{R}_j+\partial_T\Gamma[\gamma]^\ell{}_{ij}\mathcal{R}_\ell \\
&=(N\mathcal{E}_i+X^\ell\epsilon_{i\ell}{}^{k}\mathcal{H}_k)\mathcal{R}_j+\partial_T\Gamma[\gamma]^\ell{}_{ij}\mathcal{R}_\ell.
\end{align*}

\noindent Notice that there are no non-linear terms present at the level of the $L^2$ energy. However, due to the commutators $[\widehat{\mathcal{L}}_{\partial_T},\widehat{\nabla}[\gamma]_i]\E_j$ and $[\widehat{\mathcal{L}}_{\partial_T},\widehat{\nabla}[\gamma]_i]\mathcal{H}_j$, we find that at the level of the $H^1$ energy we have the non-linear terms $N\E^i\E^j\widehat{\nabla}_i\E_j$ and $N\E^i\mathcal{H}^j\widehat{\nabla}_i\mathcal{H}_j$, among others. In order to control these, we would require a uniform decay term that, as we have seen, is absent in $3+1$ dimensions. Thus we are unable to close the argument simply by means of energy estimates. This is essentially a consequence of the conformal invariance of $3+1$ Yang-Mills equations. Since the Milne model is conformal to a cylinder spacetime that exhibits no decay, one would naturally expect a loss of decay for Yang-Mills fields on a Milne background as well. To this end, the light cone estimate technique of \cite{Eardley1982B,Vazquez2022} should help closing the argument with different asymptotics for the Yang-Mills curvature.

\section{Concluding Remarks}
In this paper we gave a small data global existence result for the $n+1$, $n\geq 4$, dimensional Milne universe under coupled Einstein-Yang-Mills perturbations. One of the main difficulties is the choice of an appropriate gauge that both extracts the hyperbolic characteristics of the Einstein-Yang-Mills equations and is also suitable for long time evolution. We adapted the constant mean curvature spatial harmonic gauge introduced by \cite{Andersson2003} for the gravitational sector and a generalized Coulomb gauge introduced by \cite{Mondal2021}. A Coulomb gauge was previously utilized in the context of proving the long-time existence problem. The most notable use was by \cite{Klainerman1995} to prove the long-time existence of rough solutions of the Yang-Mills equations. However, to our knowledge, the generalized Coulomb gauge has not been utilized in previous studies. The particular reason for choosing this gauge is the fact that the divergence of the connection of non-abelian gauge theory is not a gauge covariant object, and as such one ought to subtract a reference connection in order to apply the divergence operator. This gauge choice, however, encounters the issue of a Gribov ambiguity. However, since we are working with small data, such a problem can be avoided. In fact, as mentioned previously, there does not exist a globally `good' gauge for Yang-Mills theory because of the topology and geometry of the orbit space of connections. As discussed in the main body of the article, a similar problem arises in the spatial harmonic gauge of gravity sector \cite{Fischer1996}, but, once again, a small data assumption remedies the problem.

In \cite{Friedrich1991} the Einstein-Yang-Mills equations were studied on the de-Sitter spacetime by utilizing a conformal method in $3+1$ dimensions. Recently, \cite{Liu2022} extended the result to higher dimensions. A notable difference between our study is that the de-Sitter spacetime contains a positive cosmological constant that induces a rapid (exponential) expansion. The rapid expansion of spacetime should generate sufficient decay to yield a global existence, but even though the proof of such a result is intuitive, it is a monumental task to prove it in a rigorous way. On the other hand, the Milne spacetime that we work with is a borderline case, as it exhibits a polynomial expansion instead of an exponential one. In such a scenario, it is in general difficult to obtain a uniform decay estimate for either the Einstein or Yang-Mills system. Nevertheless, \cite{Andersson2009} proved the uniform asymptotic decay property for vacuum gravity by constructing a modified energy. This worked for all dimensions $n+1$ with $n\geq 3$. In $3+1$ vacuum gravity, however, \cite{Andersson2004} provided a more geometric argument by working with the Bel-Robinson energy directly instead of the ad-hoc wave equation type of energy in \cite{Andersson2009}. Remarkably, we noted that both of these approaches fail for the Yang-Mills sector since the appropriately scaled Yang-Mills evolution equation loses the decay precisely for $3+1$ dimensions. Through a gauge-covariant argument, we showed that such a loss of decay is not an artifact of the choice of gauge. In $n+1$ dimensions for $n\geq 4$, however, one obtains the necessary decay factor in the electric part of the energy. The construction of a suitable corrected energy yields an overall uniform decay factor, allowing us to control the non-linearities. Contrary to the Yang-Mills system, this loss of decay does not cause a problem in the Maxwell case since the latter is a linear theory. To this end we note the study by \cite{Branding2019} that considers a Kaluza-Klein spacetime ${}^4M\times\mathbb{T}^q$, where ${}^4M$ is the $3+1$ dimensional Milne spacetime and $\mathbb{T}^q$ is the flat $q-$torus. A standard Kaluza-Klein reduction led to the Einstein-Maxwell-dilaton system due to $\mathbb{T}^{q}$ being abelian. In such a case, \cite{Branding2019} was able to prove the global future non-linear stability result purely by means of the energy estimates. Instead of $\mathbb{T}^q$ being the extra compact dimension, if one chooses $\mathbb{S}^q$, then the reduced system would be an Einstein-dilaton-Yang-Mills type. On the basis of our observation, a pure energy argument would not work in such a case. However, we are optimistic in the sense that we believe a stability result holds which we are simply not able to prove at present. Perhaps a more refined light-cone estimate type argument \cite{Vazquez2022} can be utilized to obtain the desired result. On the other hand, if the $3+1$ Milne model were to be truly unstable under coupled Einstein-Yang-Mills perturbations, it would lead to new physics. These issues are under intense investigation.

\bibliography{bib}
\bibliographystyle{ieeetr}

\end{document}